\documentclass[letterpaper,twocolumn,10pt]{article}
\usepackage{usenix2019_v3}
\usepackage{epsfig,endnotes}
\usepackage{cite}
\usepackage{amsmath,amssymb,amsfonts} 
\usepackage{textcomp}
\usepackage{xcolor}
\usepackage{tikz}
\usepackage{amsmath}
\usepackage{amssymb,amsfonts,amsmath}
\usepackage{textcomp}
\usepackage{csquotes}
\usepackage{booktabs}
\usepackage{color}
\usepackage{epstopdf} 
\usepackage[capitalize]{cleveref}
\usepackage{url}
\usepackage{xspace}
\usepackage{times}
\usepackage{epsfig}
\usepackage{blindtext}
\usepackage{amsmath}
\usepackage{amssymb}
\usepackage{ulem}
\usepackage{mathtools}
\usepackage{ifthen}
\usepackage{cite} 
\usepackage{graphicx}
\usepackage[capitalize]{cleveref}
\usepackage{float}
\usepackage{caption}
\usepackage{subcaption}
\usepackage{amsthm}
\usepackage{setspace}
\usepackage{enumitem} 
\usepackage{adjustbox}
\usepackage{newfloat} 
\usepackage{algorithm}
\usepackage{algorithmic}
\usepackage[algo2e]{algorithm2e} 
\usepackage{tikz}
\usepackage{amsmath}

\theoremstyle{definition}

\newtheorem{theorem}{Theorem}
\newtheorem{lemma}{Lemma}

\DeclareFloatingEnvironment[
fileext=loc,
listname={List of Algos},
name=Algorithm,
placement=pbth,
]{algo}

 \newcommand{\mempoolalgorithm}{Mandator\xspace} 

\begin{document}
\date{}

\title{\Large \bf Mandator and Sporades: Robust Wide-Area Consensus with Efficient Request Dissemination}

\author{
{\rm Pasindu Tennage}\\
EPFL
\and
{\rm Antoine Desjardins}\\
IST Austria
\and
{\rm Eleftherios Kokoris Kogias}\\
IST Austria
}

\maketitle
\begin{abstract}\label{sec:asbtr}

Consensus algorithms are deployed in the wide area to achieve high availability for geographically replicated applications. Wide-area consensus is challenging due to two main reasons: (1) low throughput due to the high latency overhead of client request dissemination and (2) network asynchrony that causes consensus protocols to lose liveness. In this paper, we propose \textbf{Mandator} and \textbf{Sporades}, a modular state machine replication algorithm that enables high performance and resiliency in the wide-area setting.

To address the high client request dissemination overhead challenge, we propose \textbf{Mandator}, a novel consensus-agnostic asynchronous dissemination layer. Mandator separates client request dissemination from the critical path of consensus to obtain high performance. Composing Mandator with Multi-Paxos (Mandator-Paxos) delivers significantly high throughput under synchronous networks. However, under asynchronous network conditions, Mandator-Paxos loses liveness which results in high latency. To achieve low latency and robustness under asynchrony, we propose \textbf{Sporades}, a novel omission fault-tolerant consensus algorithm. Sporades consists of two modes of operations -- synchronous and asynchronous -- that always ensure liveness. Sporades leverages the synchronous network to minimize the communication cost to $O(n)$ and matches the lower bound of $O(n^2)$ at worse-case executions. The combination of Mandator and Sporades (Mandator-Sporades) provides a robust and high-performing state machine replication system. 

We implement and evaluate Mandator-Sporades in a wide-area deployment running on Amazon EC2. Our evaluation shows that in the synchronous execution, Mandator-Sporades achieves 300k tx/sec throughput in less than 900ms latency, outperforming Multi-Paxos, EPaxos and Rabia by 650\% in throughput, at a modest expense of latency. Furthermore, we show that Mandator-Sporades outperforms Mandator-Paxos, Multi-Paxos, and EPaxos in the face of targeted distributed denial-of-service attacks.

\end{abstract}
\section{Introduction}\label{sec:intro}


State machine replication (SMR) \cite{cachin2011introduction} is a distributed computing abstraction that is widely used in data center-based applications \cite{burrows2006chubby, hunt2010zookeeper,maccormick2004boxwood}. SMR enables a distributed set of replicas to consistently replicate the state while remaining resilient to a minority of replica failures. At the heart of SMR lies consensus which allows replicas to agree on the total order of many values. While the early approaches to consensus focused on in-data-center deployments \cite{mazieres2007paxos, pan2021rabia, barcelona2008mencius}, recent research \cite{moraru2013there, ailijiang2017wpaxos, tennage2022baxos} has focused on the wide-area consensus algorithms. Consensus algorithms are often used in the wide-area setting to achieve high availability for applications that are geographically replicated and use the public Internet infrastructure.

 
Existing wide-area SMR protocols deliver low throughput due to high client request dissemination overhead: monolithic designs that encapsulate a batch of client requests inside consensus messages result in a bottleneck as the system throughput is then bound by the speed of consensus \cite{danezis2022narwhal}. When the consensus messages carry a batch of client requests, the asymptotic linear message complexity of consensus algorithms does not necessarily reflect the experimental performance when deployed in the wide area. This is because, in practice, the size of a consensus message is mostly influenced by the request batch size (sizes in the order of 100KB), whereas the theoretical message complexity accounts for metadata messages (sizes in the order of 100B).

 
To address the high client request dissemination overhead problem, existing research has focused on monolithic approaches that incorporate wide-area latency optimization strategies into the consensus logic. Multi-leader-based approaches \cite{charapko2021pigpaxos, zhao2018sdpaxos, marandi2010ring, rizvi2017canopus} aim at off-loading the leader request dissemination overhead to non-leader replicas. EPaxos \cite{moraru2013there} addresses the high request dissemination overhead challenge by taking advantage of application-dependent request dependency commutativity. However, we identified that existing approaches provide sub-optimal wide-area throughput, because monolithic protocols place client request dissemination in the critical path of consensus, impacting the performance when deployed in the wide area.


In this paper, we propose \textbf{Mandator}, an asynchronous and consensus-agnostic request dissemination protocol, that reliably replicates client request batches among replicas to address the high client request dissemination overhead problem. Mandator decouples client request dissemination from the total ordering, such that Mandator can run in parallel with the consensus protocol without any codependency. Mandator builds on the intuition that reliable client request dissemination is an embarrassingly parallelizable task that can run at the speed of the network (in contrast to monolithic protocols that disseminate client requests at the speed of consensus). Moreover, as we prove in Section \ref{sec:design}, reliable client request dissemination requires a simpler primitive than consensus, namely best effort broadcast with acknowledgements. Mandator leverages this best effort broadcast with acknowledgements primitive, together with a primitive that closely reflects vector clocks to build a high-performance wide-area request dissemination protocol. Consensus algorithms such as Multi-Paxos can employ Mandator as a building block and use the lightweight vector clock returned from the Mandator as the payload for consensus. With Mandator in place, consensus message sizes are lightweight and are no longer dominated by the client request batch size.


To showcase the performance advantage of Mandator, we composed Mandator with Multi-Paxos (Mandator-Paxos). In a wide-area deployment running on AWS EC2, we observed that Mandator-Paxos can deliver 300k tx/sec throughput under 900ms latency when the network is synchronous while outperforming Multi-Paxos by 650\% in throughput, at a modest expense of latency. However, in the presence of leader replica failures and network asynchrony, we observed that Mandator-Paxos provides low throughput (250k tx/sec) with significantly high latency (5s). Since Mandator guarantees liveness under asynchrony, we attribute this loss of performance under faults and asynchrony to Multi-Paxos's liveness guarantees: Multi-Paxos is a partially synchronous protocol that loses liveness under an asynchronous network and leader replica failures. Hence, as the second challenge of wide-area SMR, we focus on the loss of liveness problem under an asynchronous network and replica failures. 


Previous research addressing the liveness impact of replica failures and network asynchrony on the performance falls into two categories: (1) partially synchronous algorithms that eliminate the single leader dependency, thus minimizing the impact of partitioned replicas on the liveness of consensus, and (2) asynchronous protocols such as Ben-Or \cite{ben1983another} that guarantee liveness under an asynchronous network. Partially synchronous protocols such as EPaxos \cite{moraru2013there} and Mencius \cite{barcelona2008mencius} avoid the single leader bottleneck, and are immune to leader network partitions, but they lose liveness under an asynchronous network. Asynchronous protocols such as Ben-Or \cite{ben1983another} are immune to asynchronous network conditions, but have quadratic message complexity and low performance in the typical synchronous case. As a result, asynchronous protocols are not typically deployed in the wide-area\footnote{Rabia \cite{pan2021rabia} is a protocol built on top of Ben-Or \cite{ben1983another}, but Rabia operates in a synchronous network and loses liveness in an asynchronous network}. We believe that a practical wide-area consensus algorithm should have the best of both worlds (1) high performance under synchrony, and (2) high resilience and liveness under asynchronous networks and replica failures. 


To address the challenge of liveness under network asynchrony and replica failures in SMR, we propose \textbf{Sporades}, a novel omission fault-tolerant consensus algorithm, that combines an optimistic, linear overhead synchronous path protocol with a pessimistic path protocol that guarantees liveness under asynchrony. Sporades, has two modes; (1) a synchronous mode and (2) an asynchronous mode. The synchronous mode of Sporades commits a batch of client requests in a single network round-trip when the network is synchronous. However, when the network is asynchronous, Sporades falls back to an asynchronous mode that commits client requests with quadratic message complexity. This makes the Sporades guarantee liveness in the face of replica failures, asynchronous networks and DDoS attacks. The composition of Mandator with Sporades (Mandator-Sporades) is a novel SMR algorithm that guarantees high performance in the wide area while providing liveness guarantees under asynchronous networks and replica failures.
 
We make three contributions in this paper: 
\vspace{-1mm}
\begin{itemize}
    \vspace{-2mm}
    \item We design and implement Mandator: a novel consensus-agnostic request dissemination protocol. We implement Multi-Paxos on top of Mandator, which we refer to as Mandator-Paxos
    \vspace{-2mm}
    \item We design Sporades, a novel omission fault tolerant consensus algorithm that guarantees liveness under asynchronous network conditions. We implement Sporades on top of Mandator (Mandator-Sporades)
    \vspace{-2mm}
    \item We experimentally show that Mandator-Paxos and Mandator-Sporades achieve 300k tx/sec throughput under 900ms median latency, compared to EPaxos \cite{moraru2013there} (6.5k tx/s under 720ms), Multi-Paxos \cite{mazieres2007paxos} (40k under 295ms) and Rabia \cite{pan2021rabia} (0.5k under 500ms) in the synchronous network conditions in a wide-area deployment running on AWS EC2. We also show that Mandator-Sporades outperforms Mandator-Paxos by 60\% in throughput in the face of DDoS attacks.
    \vspace{-2mm}
\end{itemize}
\section{Preliminaries}\label{sec:back}


This section provides an overview of the system model and SMR.

\subsection{System Model}


We consider a system with $n$ replicas. We assume omission fault-tolerant replicas. An omission fault occurs when a replica does not send (or receive) a message that it is supposed to send (or receive) according to the specification of the algorithm \cite{cachin2011introduction}. Up to $f$ (where $n$ $\geq$ 2$f$ + 1) number of replicas can display omission-fault behaviors, but replicas do not equivocate.


We assume perfect point-to-point links between each pair of replicas; messages from any correct replica p\textsubscript{i} to any correct replica p\textsubscript{j} are eventually delivered. The replicas are connected in a logical-complete graph. We say that a replica broadcasts a message $m$ if it sends $m$ to all $n$ replicas. We assume an adversarial network that can arbitrarily reorder and delay the messages.     


Let $\Delta$ be the upper bound on message transmission delay and GST be the global stabilization time. An execution of a protocol is considered synchronous if each message sent from a correct replica p\textsubscript{i} is delivered by replica p\textsubscript{j} within $\Delta$. An execution of a protocol is considered asynchronous if there exists no time bound $\Delta$ for message delivery. An execution of a protocol is said to be partially synchronous, if there is an unknown GST such that once GST is reached, each message sent by process p\textsubscript{i} is delivered by process p\textsubscript{j} within a known $\Delta$\cite{dwork1988consensus}.


Due to the FLP impossibility result~\cite{fischer1985impossibility}, any deterministic algorithm cannot solve consensus under asynchrony even under a single replica failure. In Sporades, we circumvent the FLP impossibility result using partial synchrony and randomization. In a typical wide-area network, there are periods in which the network behaves synchronously, followed by phases where the network shows asynchronous behaviour. Sporades makes use of this network behaviour and operates in two modes: (1) synchronous mode and (2) asynchronous mode. During the synchronous periods, Sporades assumes partial synchrony for liveness and during the asynchronous mode, Sporades assumes an asynchronous network and uses randomization.

\subsection{Consensus and State Machine Replication}

Consensus is a distributed system abstraction that enables a set of replicas to reach an agreement on a single value. A correct consensus algorithm satisfies four properties \cite{cachin2011introduction}: (1) \textit{validity}: the agreed upon value should be previously proposed by a replica, (2) \textit{termination}: every correct process eventually decides some value, (3) \textit{integrity}: no process decides twice and (4) \textit{agreement}: no two correct processes decide differently.


SMR is a use-case of consensus, where replicas run multiple instances of the consensus algorithm to agree on a series of values \cite{cachin2011introduction}\cite{antoniadis2018state}. In this paper, we solve the SMR problem. A correct SMR algorithm satisfies two properties; (1) \textit{safety}: no two replicas commit different client requests for the same log position and (2) \textit{liveness}: each client request is eventually committed. We assume that each client request will be repeatedly proposed by replicas until it is committed.


We focus on the chaining approach to SMR similar to Raft \cite{ongaro2014search}, in which each new proposal to SMR has a reference (parent link) to the previous proposal, and each commit operation commits the entire history of client requests.
\section{Design}\label{sec:design}


In this section, we discuss the design of Mandator and Sporades as two building blocks and then present how to build a high-performing and resilient wide-area SMR algorithm by combining the Mandator and Sporades in section \ref{sec:impl}.

\subsection{\mempoolalgorithm}


The Mandator is a consensus-agnostic reliable request dissemination layer whose goal is to reliably disseminate client requests to be voted for in consensus. Mandator does not solve consensus; rather, it ensures that a majority of the replicas are aware of each client request batch that is later proposed in the consensus. As a result, the consensus layer does not have to wait for client request dissemination (because the Mandator can be executed ahead of time concurrently). The consensus layer can refer to request batches by their unique identifiers (rather than broadcasting the entire batch of client requests), making communication lighter and faster for consensus.


Because the Mandator algorithm has each replica running its thread as a leader, the workload for all replicas is the same. Algorithm \ref{algo:mempoolalgorithm} depicts the pseudo-code of Mandator. Each replica broadcasts a <new-Mandator-batch> $B_1$, then waits for $n-f$ <Mandator-votes> on $B_1$ before broadcasting the next <new-Mandator-batch> $B_2$. When $B_2$ is received, all processes understand that $B_1$ has received $n-f$ Mandator-votes. As a result, this algorithm is optimized for sending Mandator-batches in a chain. It is asynchronous and runs in parallel with the consensus layer.

\textbf{Terminology}: 
We introduce \textbf{Mandator-batch}. A Mandator-batch contains four fields: (1) round number, (2) reference to the parent Mandator-batch, (3) one or more client requests and (4) the unique identifier that uniquely identifies the Mandator-batch. Since the unique identifier is implicit, we only use the first three fields in the following description.

Mandator provides the following generic interface.
\vspace{-2mm}

\begin{itemize}
    \item \textbf{write($B$)}: write a new Mandator-batch. We say that the write($B$) is successful if the node that writes $B$ receives at least $n-f$ <Mandator-votes> for $B$. \vspace{-2mm}
    \item \textbf{read($B$)}: read the Mandator-batch corresponding to the identifier of $B$. A read is successful if it returns a Mandator-batch $B$=($r$, $B.parent$, $cmds$) which was successfully written previously using write($B$) \vspace{-2mm} 
    \item \textbf{read\_causal($B$)}: read the set of Mandator-batches that causally precede $B$. read\_causal($B$) succeeds if read($B$) succeeds and if all the Mandator-batches that were written before write($B$) with which B has a causal dependency are in the returned Mandator-batch list.
\end{itemize}

The  Mandator algorithm has two properties: \vspace{-2mm}

\begin{itemize}
    \item \textbf{Availability}: if read($B$) is invoked after completing write($B$), then read($B$) eventually returns $B$ \vspace{-2mm}
    \item  \textbf{Causality}: a successful read\_causal($B$) returns all the Mandator-batches with which B has a causal dependency \vspace{-2mm}
\end{itemize}

\begin{figure}
    \centering
    \scalebox{.85}{\input{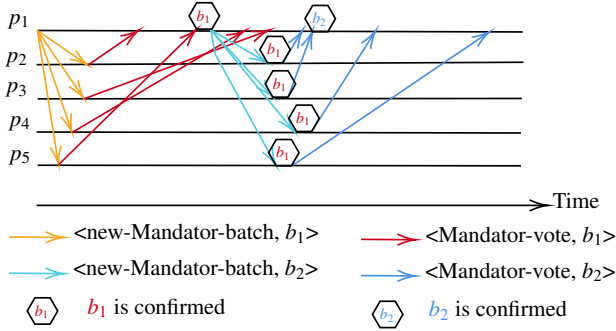}}
    \caption{An execution of Mandator with 5 replicas. All replicas act as leaders, simultaneously. For clarity we only show the execution with only $p_1$ as the leader}
    \vspace{-5mm}
    \label{fig:mem-pool}
\end{figure}

\textbf{Typical execution} The Mandator algorithm runs asynchronously. All replicas run the same algorithmic steps concurrently and do not need to wait for the progress of another replica to move on. This characteristic is what allows the Mandator algorithm to yield a much higher throughput than consensus. The execution we explain in the following is depicted in figure~\ref{fig:mem-pool}.

The sender is process $p_1$. The other processes are $p_2$, $p_3$, $p_4$ and $p_5$; there are only $5$ processes in total, which means the threshold $f$ is $2$. $p_1$ has received a batch of requests from clients (line~\ref{M:line:7}) and broadcast them as a <new-Mandator-batch> $B_1$ with round number $r_1$ (line~\ref{M:line:13}). $p_1$, $p_2$, $p_3$, $p_4$ and $p_5$ eventually receive $B_1$ (line~\ref{M:line:14}) and send back a <Mandator-vote> to $p_1$ (line~\ref{M:line:17}). Here $p_3$ and $p_4$ are slow and send their vote the latest. $p_1$ receives $p_2$'s <Mandator-vote> first, then $p_5$'s. The <Mandator-votes> from $p_1$, $p_2$ and $p_5$ represent $n-f=5-2=3$ votes, so $p_1$ can consider the round $r_1$ completed (line~\ref{M:line:18}).

$p_1$ receives another batch of requests from clients and broadcasts a <new-Mandator-batch> $B_2$ to all replicas (line~\ref{M:line:13}). $B_2$ contains the $lastCompletedRound$[$p_1$] which indicates that $B_1$ was completed. Upon receiving $B_2$, $p_2$, $p_3$, $p_4$ and $p_5$ know that $B_1$ was completed and update their own $lastCompletedRounds$[$p_1$] to $B_1$ (line~\ref{M:line:16}). They will send a <Mandator-vote> for $B_2$ to $p_1$ (line~\ref{M:line:17}) and the algorithm continues.

\textbf{Complexity} The Mandator Algorithm has a linear complexity: for each batch of client requests $cl$, one Mandator-batch is broadcast to all replicas and each of these replicas replies to the sender with a <Mandator-vote>.

\textbf{Proof sketch}: The proofs of availability and causality directly follow from the use of quorum-based broadcast with vote-acknowledgements and the chaining of Mandator-batches. A formal proof appears in Appendix \ref{sec:formal-proof-mandator}.

\textbf{Interaction with Consensus}: Mandator provides the getClientRequests() interface to the consensus layer. getClientRequests() returns the lastCompletedRounds[] (see line~\ref{M:line:2} in algorithm~\ref{algo:mempoolalgorithm}) which contains the last completed Mandator-Batch for each replica. Since the lastCompletedRounds is an integer array of N elements, the consensus blocks become lightweight. In contrast, in monolithic protocols such as Multi-Paxos, the consensus messages have to carry the entire batch of client requests, thus delivering lower performance compared to Multi-Paxos built on top of Mandator.

We integrated Mandator with Multi-Paxos (Mandator-Paxos) and observed that Mandator-Paxos delivers 300k tx/sec throughput under 900ms median latency in a wide-area deployment running on AWS EC2. In contrast, Multi-Paxos delivers 40k tx/sec under 295ms. However, we observed that under DDoS attacks, the performance of Mandator-Paxos decreases to 250 tx/sec under 5s median latency. Since Mandator guarantees liveness under asynchronous network conditions, we attribute this low performance under DDoS attacks to Multi-Paxos's liveness guarantees: Multi-Paxos is a partially synchronous protocol that loses liveness under network asynchrony. We believe that a robust wide-area SMR system should deliver liveness even under network asynchrony. To address this liveness issue in the consensus layer, in the next sub-section we propose a new omission fault-tolerant consensus algorithm, Sporades, that guarantees liveness under DDoS attacks and network asynchrony.
 
\begin{algorithm}[h!] \small
\caption{Mandator Algorithm for process $p_i$, $i\in 0..n-1$} \label{algo:mempoolalgorithm}
\begin{algorithmic}[1]
\STATE \textbf{Local State:}
\begin{ALC@g}
\STATE lastCompletedRounds[], an array of N elements (with N the number of replicas) that keeps track of the last Mandator-batch for which at least $n-f$ Mandator-votes were collected for each replica \label{M:line:2}
\STATE chains[][], a 2D array that saves the $i$\textsuperscript{th} Mandator-batch created by replica $j$
\STATE buffer, a queue storing incoming client requests
\STATE awaitingAcks, a boolean variable which states whether this replica is waiting for Mandator-votes. Initialized to \FALSE
\newline
\end{ALC@g}

\REQUIRE maximum batch time and batch size
\newline

\STATE \textbf{Upon} receiving a batch of client requests $cl$: \label{M:line:7}
\begin{ALC@g}
\STATE push $cl$ to buffer \label{M:line:8}
\newline
\end{ALC@g}

\STATE \textbf{Upon} (size of incoming buffer reaching batch size \OR maximum batch time is passed) \AND awaitingAcks is \FALSE:
\begin{ALC@g}
\STATE set B\textsubscript{parent} be the Mandator-batch corresponding to chains[i][lastCompletedRounds[i]], i.e. B\textsubscript{parent} is the last Mandator-batch proposed by $p_i$ that received $n-f$ Mandator-votes
\STATE create new Mandator-batch \\ \hfill B = (lastCompletedRounds[i]+1, B\textsubscript{parent}, buffer.popAll())
\STATE Set isAwaiting to \TRUE
\STATE broadcast $<$new-Mandator-Batch, B$>$ \label{M:line:13}
\newline
\end{ALC@g}

\STATE \textbf{Upon} receiving $<$new-Mandator-batch, B$>$ from $p_j$ \label{M:line:14}
\begin{ALC@g}
\STATE set chains[j][B.round] to B
\STATE set lastCompletedRounds[j] to B.parent.round \label{M:line:16}
\STATE send $<$Mandator-vote, B.round$>$ to $p_j$ \label{M:line:17}
\newline
\end{ALC@g}

\STATE \textbf{Upon} receiving $n-f$ $<$Mandator-vote, r$>$ for the same r \AND \\ \hfill r = lastCompletedRounds[i]+1 \AND awaitingAcks is \TRUE \label{M:line:18}
\begin{ALC@g}
\STATE set awaitingAcks to \FALSE
\STATE set lastCompletedRounds[i]+=1
\newline
\end{ALC@g}

\STATE \textbf{procedure getClientRequests()} \label{M:line:21}
\begin{ALC@g}
\RETURN lastCompletedRounds \hspace{1cm}
\newline
\end{ALC@g}

\STATE \textbf{Upon} onCommit (r[])
\begin{ALC@g}
\FOR{k in range 0, N} 
\STATE commit the causal history of chains[k][r[k]]  
\ENDFOR
\end{ALC@g}
\end{algorithmic}
\end{algorithm}
\vspace{-2mm}

\subsection{Sporades}\label{subsec:sporades}


We propose Sporades, a novel omission fault-tolerant consensus algorithm that guarantees liveness under network asynchrony and DDoS attacks. In particular, Sporades has two paths; the synchronous path that commits client requests in a single round trip and an asynchronous path that has quadratic message complexity that commits client requests under network asynchrony. Sporades dynamically switches between the synchronous and the asynchronous paths depending on the network condition.

\subsubsection{Terminology}

\textbf{View number} and \textbf{Round number}: Sporades progresses as a sequence of views $v$ and rounds $r$ where each view has one or more rounds. A view represents the term of a leader while a round represents the successive phases of that term. The pair $(v,r)$ is called a \textbf{rank}.

\textbf{Block Format}: There are two kinds of Sporades blocks: (1) \textbf{synchronous blocks} and (2) \textbf{asynchronous blocks}. Both types of blocks consist of five elements: (1) batch of client requests, (2) view number, (3) round number, (4) parent link to a block with a lower rank, and (5) level. The rank of a block is ($v, r$) and blocks are compared lexicographically by their rank: first by the view number, then by the round number. The blocks are connected in a chain using the parent links. We denote that block A extends block B if there exist a set of blocks $b\textsubscript{1}, b\textsubscript{2}, b\textsubscript{3}, .. b\textsubscript{k}$ such that there exists a parent link from $b\textsubscript{i}$ to $b\textsubscript{i-1}$ $\forall$ $i$ in $range(2, k)$ and $b\textsubscript{1}$ = $B$ and $b\textsubscript{k}$ = $A$. The level element of the block refers to the asynchronous level (can take either the value $1$ or $2$). For the synchronous blocks, the level is always $-1$.

\textbf{Common-coin-flip($v$)}: We employ a common-coin-flip primitive as a building block for designing the Sporades. For each value of $v$, common-coin-flip($v$) returns a positive integer in the range (0, n-1) where n is the total number of replicas. The common-coin-flip($v$) satisfies two properties; (1) each invocation of common-coin-flip($v$) for view $v$ at each replica should return the same integer value, if the same $v$ is used as input, and (2) output of the invocation of common-coin-flip($i$) should be independent from common-coin-flip($j$) for j $\neq$ i.


For the implementation of the common-coin-flip($v$), we use the approach used by Rabia\cite{pan2021rabia}: we use a pseudo-random number generator with the same (shared secret) seed at each replica and pre-generate random numbers for each view number. Then the invocation of common-coin-flip($v$) returns the preset random number corresponding to $v$. Since the nodes are non-byzantine, and since the network adversary cannot view the state of the replicas, this implementation of common-coin-flip($v$) delivers the two properties mentioned above.

\subsubsection{Sporades Algorithm}\label{sec:subsubsporadesdescription}

Sporades consists of two sub-protocols: the \textbf{synchronous} mode that works under synchrony and the \textbf{asynchronous} mode that works under asynchrony. Sporades dynamically switches between the two modes depending on the network condition. Algorithms \ref{algo:fallbacksteadyalgorithm} and \ref{algo:fallbackalgorithm} depict the pseudo-codes of synchronous and asynchronous modes of Sporades, respectively.

\begin{figure}
    \centering
    \scalebox{.67}{\input{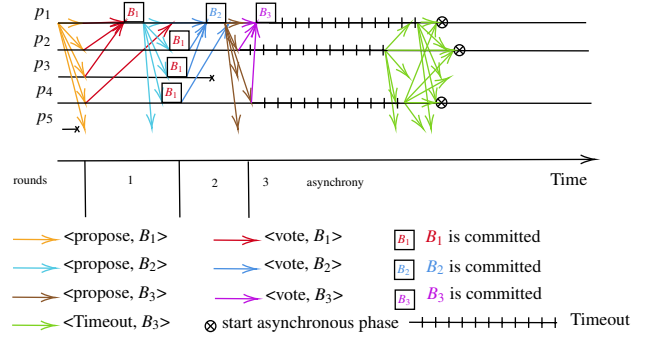}}
    \caption{An execution of synchronous mode of Sporades with 5 replicas}
    \vspace{-5mm}
    \label{fig:consensus-regular}
\end{figure}


The synchronous mode of Sporades is a leader-based consensus algorithm, as depicted in figure \ref{fig:consensus-regular}. The leader for each view is predetermined and shared between peers on bootstrap. The leader replica broadcasts a <propose> message for a given new block $B_p$ containing a rank $(v,r)$ and the reference of the last committed block $B_c$ (line~\ref{S:line:18} of algorithm~\ref{algo:fallbacksteadyalgorithm}). Each process $p$ delivers the <propose> message, if the rank of $B_p$ is greater than the rank of $p$ and if $p$ is in the synchronous mode of operation (see line~\ref{S:line:20} in algorithm~\ref{algo:fallbacksteadyalgorithm}). If these two conditions are met, then each replica commits the block $B_c$ (and the causal history) and sends their <vote> for $B_p$ to the leader replica. Upon receiving $n-f$ <vote> messages for $B_p$, the leader replica commits $B_p$ (and the causal history), forms a new block to be proposed $B_p2$ with an incremented round number compared to $B_p$ and broadcasts a <propose> message with the information that $B_p$ has been committed (see lines~\ref{S:line:9}-\ref{S:line:18} of algorithm~\ref{algo:fallbacksteadyalgorithm}).


Each replica has a Timeout clock which is reset whenever the replica receives a new <propose> message (line~\ref{S:line:26} of Algorithm \ref{algo:fallbacksteadyalgorithm}). If the timeout expires, however, they will broadcast a <timeout> message containing the highest block they are aware of (see line~\ref{S:line:28} of algorithm~\ref{algo:fallbacksteadyalgorithm}).

\begin{figure}
    \centering
    \scalebox{.54}{\input{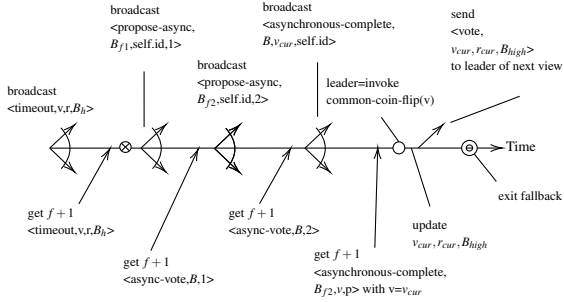}}
    \caption{An execution of the asynchronous mode of Sporades from the perspective of a single process}
    \vspace{-5mm}
    \label{fig:consensus-fallback}
\end{figure}


Upon receiving $n-f$ <timeout> messages, Sporades enters the asynchronous mode of operation as depicted in figure \ref{fig:consensus-fallback}. In this mode, all replicas act as leaders, concurrently. Each replica takes the highest block they are aware of, forms a height $1$ asynchronous block with a monotonically increasing rank compared to the highest block it received and sends a <propose-async> message (line~\ref{A:line:1}-\ref{A:line:7} in algorithm~\ref{algo:fallbackalgorithm}). Each replica sends back a <vote-async> message to the sending process if the rank of the proposed height $1$ block is greater than the highest ranked block witnessed so far (line~\ref{A:line:8}-\ref{A:line:10} of algorithm~\ref{algo:fallbackalgorithm}). Upon receiving $n-f$  <vote-async> messages for the height $1$ asynchronous block, each replica will send a height $2$ asynchronous fallback block. The algorithm allows catching up to higher ranked blocks by building upon another process's height $1$ block. This is meant to ensure liveness for replicas that fall behind. Once $n-f$ <vote-async>s have been gathered for the height $2$ asynchronous block by each replica $p$, the asynchronous mode is considered complete and each replica $p$ broadcasts an <asynchronous-complete> message. When $n-f$ replicas have broadcast <asynchronous-complete> messages, all replicas exit the asynchronous protocol by flipping a common-coin and committing the height $2$ block (and the causal history) from the process designated by the common-coin-flip($v$). Because of the potential $f$ faults, each replica commits a block $B$ if $B$ arrived amongst the first $n-f$ <asynchronous-complete> blocks. After that, replicas exit the asynchronous mode and resume the synchronous path by unicasting a vote to the leader of the next view for the block block\textsubscript{high}.

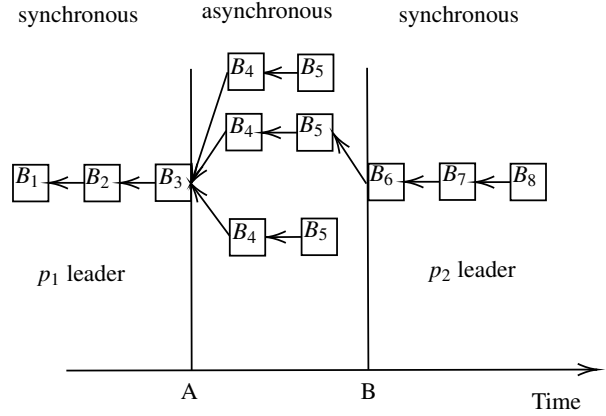
\begin{figure}
    \centering
    \scalebox{.9}{\tikzset{every picture/.style={line width=0.75pt}} 

\begin{tikzpicture}[x=0.75pt,y=0.75pt,yscale=-1,xscale=1]

\draw  [fill={rgb, 255:red, 255; green, 255; blue, 255 }  ,fill opacity=1 ] (264.75,95.97) -- (284.62,95.97) -- (284.62,116.53) -- (264.75,116.53) -- cycle ;
\draw   (225.2,124.34) -- (245.07,124.34) -- (245.07,144.89) -- (225.2,144.89) -- cycle ;
\draw    (225.2,134.68) -- (207.13,134.68) ;
\draw [shift={(205.13,134.68)}, rotate = 360] [color={rgb, 255:red, 0; green, 0; blue, 0 }  ][line width=0.75]    (10.93,-3.29) .. controls (6.95,-1.4) and (3.31,-0.3) .. (0,0) .. controls (3.31,0.3) and (6.95,1.4) .. (10.93,3.29)   ;
\draw  [fill={rgb, 255:red, 255; green, 255; blue, 255 }  ,fill opacity=1 ] (304.49,96.18) -- (324.36,96.18) -- (324.36,116.73) -- (304.49,116.73) -- cycle ;
\draw    (304.49,106.52) -- (286.42,106.52) ;
\draw [shift={(284.42,106.52)}, rotate = 360] [color={rgb, 255:red, 0; green, 0; blue, 0 }  ][line width=0.75]    (10.93,-3.29) .. controls (6.95,-1.4) and (3.31,-0.3) .. (0,0) .. controls (3.31,0.3) and (6.95,1.4) .. (10.93,3.29)   ;
\draw   (185.26,124.54) -- (205.13,124.54) -- (205.13,145.1) -- (185.26,145.1) -- cycle ;
\draw [fill={rgb, 255:red, 208; green, 2; blue, 27 }  ,fill opacity=1 ]   (264.75,105.49) -- (246.19,133.22) ;
\draw [shift={(245.07,134.88)}, rotate = 303.79] [color={rgb, 255:red, 0; green, 0; blue, 0 }  ][line width=0.75]    (10.93,-3.29) .. controls (6.95,-1.4) and (3.31,-0.3) .. (0,0) .. controls (3.31,0.3) and (6.95,1.4) .. (10.93,3.29)   ;
\draw    (175.17,239.97) -- (470.68,239.4) ;
\draw [shift={(472.68,239.4)}, rotate = 179.89] [color={rgb, 255:red, 0; green, 0; blue, 0 }  ][line width=0.75]    (10.93,-3.29) .. controls (6.95,-1.4) and (3.31,-0.3) .. (0,0) .. controls (3.31,0.3) and (6.95,1.4) .. (10.93,3.29)   ;
\draw    (343.29,135.13) -- (326.1,109) ;
\draw [shift={(325,107.33)}, rotate = 56.65] [color={rgb, 255:red, 0; green, 0; blue, 0 }  ][line width=0.75]    (10.93,-3.29) .. controls (6.95,-1.4) and (3.31,-0.3) .. (0,0) .. controls (3.31,0.3) and (6.95,1.4) .. (10.93,3.29)   ;
\draw    (343.7,70.17) -- (344.43,239.53) ;
\draw    (245,70.77) -- (245.4,239.97) ;
\draw   (384.4,123.74) -- (404.27,123.74) -- (404.27,144.29) -- (384.4,144.29) -- cycle ;
\draw    (384.4,134.08) -- (366.33,134.08) ;
\draw [shift={(364.33,134.08)}, rotate = 360] [color={rgb, 255:red, 0; green, 0; blue, 0 }  ][line width=0.75]    (10.93,-3.29) .. controls (6.95,-1.4) and (3.31,-0.3) .. (0,0) .. controls (3.31,0.3) and (6.95,1.4) .. (10.93,3.29)   ;
\draw   (424.08,123.94) -- (443.95,123.94) -- (443.95,144.5) -- (424.08,144.5) -- cycle ;
\draw    (424.08,134.28) -- (406.01,134.28) ;
\draw [shift={(404.01,134.28)}, rotate = 360] [color={rgb, 255:red, 0; green, 0; blue, 0 }  ][line width=0.75]    (10.93,-3.29) .. controls (6.95,-1.4) and (3.31,-0.3) .. (0,0) .. controls (3.31,0.3) and (6.95,1.4) .. (10.93,3.29)   ;
\draw   (344.46,123.94) -- (364.33,123.94) -- (364.33,144.5) -- (344.46,144.5) -- cycle ;
\draw  [fill={rgb, 255:red, 255; green, 255; blue, 255 }  ,fill opacity=1 ] (304.99,62.68) -- (324.86,62.68) -- (324.86,83.23) -- (304.99,83.23) -- cycle ;
\draw    (304.99,73.02) -- (286.92,73.02) ;
\draw [shift={(284.92,73.02)}, rotate = 360] [color={rgb, 255:red, 0; green, 0; blue, 0 }  ][line width=0.75]    (10.93,-3.29) .. controls (6.95,-1.4) and (3.31,-0.3) .. (0,0) .. controls (3.31,0.3) and (6.95,1.4) .. (10.93,3.29)   ;
\draw  [fill={rgb, 255:red, 255; green, 255; blue, 255 }  ,fill opacity=1 ] (265.95,62.07) -- (285.82,62.07) -- (285.82,82.63) -- (265.95,82.63) -- cycle ;
\draw [fill={rgb, 255:red, 208; green, 2; blue, 27 }  ,fill opacity=1 ]   (265.48,72.68) -- (245.7,132.98) ;
\draw [shift={(245.07,134.88)}, rotate = 288.16] [color={rgb, 255:red, 0; green, 0; blue, 0 }  ][line width=0.75]    (10.93,-3.29) .. controls (6.95,-1.4) and (3.31,-0.3) .. (0,0) .. controls (3.31,0.3) and (6.95,1.4) .. (10.93,3.29)   ;
\draw  [fill={rgb, 255:red, 255; green, 255; blue, 255 }  ,fill opacity=1 ] (306.99,154.68) -- (326.86,154.68) -- (326.86,175.23) -- (306.99,175.23) -- cycle ;
\draw    (306.99,165.02) -- (289.2,165.14) ;
\draw [shift={(287.2,165.15)}, rotate = 359.62] [color={rgb, 255:red, 0; green, 0; blue, 0 }  ][line width=0.75]    (10.93,-3.29) .. controls (6.95,-1.4) and (3.31,-0.3) .. (0,0) .. controls (3.31,0.3) and (6.95,1.4) .. (10.93,3.29)   ;
\draw  [fill={rgb, 255:red, 255; green, 255; blue, 255 }  ,fill opacity=1 ] (266.75,154.97) -- (286.62,154.97) -- (286.62,175.53) -- (266.75,175.53) -- cycle ;
\draw [fill={rgb, 255:red, 208; green, 2; blue, 27 }  ,fill opacity=1 ]   (266.75,164.49) -- (246,136.49) ;
\draw [shift={(244.81,134.88)}, rotate = 53.46] [color={rgb, 255:red, 0; green, 0; blue, 0 }  ][line width=0.75]    (10.93,-3.29) .. controls (6.95,-1.4) and (3.31,-0.3) .. (0,0) .. controls (3.31,0.3) and (6.95,1.4) .. (10.93,3.29)   ;
\draw    (185.2,134.68) -- (167.13,134.68) ;
\draw [shift={(165.13,134.68)}, rotate = 360] [color={rgb, 255:red, 0; green, 0; blue, 0 }  ][line width=0.75]    (10.93,-3.29) .. controls (6.95,-1.4) and (3.31,-0.3) .. (0,0) .. controls (3.31,0.3) and (6.95,1.4) .. (10.93,3.29)   ;
\draw   (145.26,124.54) -- (165.13,124.54) -- (165.13,145.1) -- (145.26,145.1) -- cycle ;

\draw (434.47,251.17) node [anchor=north west][inner sep=0.75pt]   [align=left] {Time};
\draw (237.9,245.23) node [anchor=north west][inner sep=0.75pt]   [align=left] {A};
\draw (338.77,245.5) node [anchor=north west][inner sep=0.75pt]   [align=left] {B};
\draw (296.4,39.13) node   [align=left] {\begin{minipage}[lt]{68pt}\setlength\topsep{0pt}
asynchronous
\end{minipage}};
\draw (146.67,33.67) node [anchor=north west][inner sep=0.75pt]   [align=left] {synchronous};
\draw (359.97,33.87) node [anchor=north west][inner sep=0.75pt]   [align=left] {synchronous};
\draw (157.33,179.33) node [anchor=north west][inner sep=0.75pt]   [align=left] {$\displaystyle p_{1}$ leader};
\draw (376.33,177) node [anchor=north west][inner sep=0.75pt]   [align=left] {$\displaystyle p_{2}$ leader};
\draw (265.67,63.33) node [anchor=north west][inner sep=0.75pt]   [align=left] {$\displaystyle B_{4}$};
\draw (265.33,96.67) node [anchor=north west][inner sep=0.75pt]   [align=left] {$\displaystyle B_{4}$};
\draw (267,156) node [anchor=north west][inner sep=0.75pt]   [align=left] {$\displaystyle B_{4}$};
\draw (225.67,125.33) node [anchor=north west][inner sep=0.75pt]   [align=left] {$\displaystyle B_{3}$};
\draw (184.67,125.33) node [anchor=north west][inner sep=0.75pt]   [align=left] {$\displaystyle B_{2}$};
\draw (145.33,124.67) node [anchor=north west][inner sep=0.75pt]   [align=left] {$\displaystyle B_{1}$};
\draw (305,63.33) node [anchor=north west][inner sep=0.75pt]   [align=left] {$\displaystyle B_{5}$};
\draw (304.33,96) node [anchor=north west][inner sep=0.75pt]   [align=left] {$\displaystyle B_{5}$};
\draw (307,154.67) node [anchor=north west][inner sep=0.75pt]   [align=left] {$\displaystyle B_{5}$};
\draw (343.67,123.67) node [anchor=north west][inner sep=0.75pt]   [align=left] {$\displaystyle B_{6}$};
\draw (384.67,123.67) node [anchor=north west][inner sep=0.75pt]   [align=left] {$\displaystyle B_{7}$};
\draw (424,124.33) node [anchor=north west][inner sep=0.75pt]   [align=left] {$\displaystyle B_{8}$};

\end{tikzpicture}}
    \caption{An execution of Sporades from a replicated log perspective}
    \vspace{-5mm}
    \label{fig:Blockchain}
\end{figure}


Figure~\ref{fig:Blockchain} features an example of the overall process of going from synchronous to asynchronous to synchronous again in Sporades. After committing 3 blocks, $B_1$, $B_2$ and $B_3$ in the synchronous mode, all replicas go to the asynchronous protocol. There, all of the alive replicas, namely $p_1$, $p_2$ and $p_4$ propose a height $1$ block $B_4$ then a height 2 block $B_5$. The common coin flip designates $p_2$ as the elected leader of this asynchronous path. Since $p_2$'s <asynchronous-complete> messages were part of the first $n-f=3$ first such messages received by all correct replicas , all replicas commit $B_5$ proposed by $p_2$ (and the causal history) and in the following synchronous execution commit $B_6$, $B_7$ and $B_8$.

To illustrate the operation of Sporades, in the Section \ref{sec:sporades-execution} in Appendix, we give an example execution of Sporades.


We now give some intuitions of the proofs for Sporades. The complete and formal proofs have been deferred to the appendix, section~\ref{sec:formal-proof}.

\textbf{Sketch proof of safety} 
The safety follows from the fact that if a block $B$ is committed in a round $r$, then all the blocks with round $r' \geq r$ will extend $B$. In the synchronous path of Sporades, a block is committed by the leader node when a majority of nodes vote for the synchronous block in the same round. Hence, if the leader commits a block $B$, then it is guaranteed that at least a majority of the replicas have $B$ as its $block\textsubscript{high}$. In the asynchronous path, a block $B$ is committed if it is among the first $n-f$ <asynchronous-complete> messages received, which implies that a majority of the replicas have received $B$. Hence if at least one node commits $B$ in the asynchronous path, then it is guaranteed that at least a majority of the nodes set $B$ as block\textsubscript{high}, thus extending $B$ in the next round.

\textbf{Sketch proof of liveness} 
During the synchronous phase, transactions get committed at each round by the leader, after the leader receives votes from a quorum of replicas. The asynchronous phase is more delicate to address. The essential feature is that, if more than $f$ processes enter the asynchronous path, i.e. there are less than $n-f$ processes on the synchronous path, i.e. the synchronous path cannot progress anymore, then all replicas eventually enter the asynchronous phase. The asynchronous phase (with at least $n-f$ correct replicas in it) will eventually reach a point where $n-f$ correct replicas have sent <asynchronous-complete> messages. The common-coin-flip($v$) will therefore take place. If the result of the common-coin-flip($v$) lands on one of the first $n-f$ replicas to have submitted an <asynchronous-complete> message, then a new block is committed. Since the coin is unbiased, there is a probability $p>\frac{1}{2}$ that a new block is committed. Therefore, the algorithm will, with high probability, keep committing blocks in a finite time.

\textbf{Complexity} The synchronous path of Sporades has a linear message and bit complexity for committing a block. Indeed, the leader broadcasts one fixed-size message to all processes, and each of them will send the leader a fixed-size response. The asynchronous path of Sporades has a complexity of $O(n^{2})$. Each process $p$ acts as a leader, first broadcasting $n$ <propose-async> messages of height $1$. Then each correct process (at least $n-f$) sends back to $p$ a <vote-async> message of height $1$. $p$ then broadcast another $n$ <propose-async> messages of height $2$. At most $n-f$ processes answer to $p$ with a <vote-async> message of height $2$. Upon receiving $n-f$ <vote-async> messages of height $2$, $p$ will broadcast $n$ <asynchronous-complete> messages. Thus for each process, at most $O(n)$ messages are exchanged during the asynchronous phase. Therefore the complexity is $O(n^{2})$.

\begin{figure}
    \centering
    \scalebox{.6}{\tikzset{every picture/.style={line width=0.75pt}} 

\begin{tikzpicture}[x=0.75pt,y=0.75pt,yscale=-1,xscale=1]

\draw   (69,42) -- (117,42) -- (117,85) -- (69,85) -- cycle ;
\draw   (69,94) -- (117,94) -- (117,137) -- (69,137) -- cycle ;
\draw   (69,151) -- (117,151) -- (117,194) -- (69,194) -- cycle ;
\draw   (69,206) -- (117,206) -- (117,249) -- (69,249) -- cycle ;
\draw   (69,261) -- (117,261) -- (117,304) -- (69,304) -- cycle ;
\draw   (149,43) -- (197,43) -- (197,86) -- (149,86) -- cycle ;
\draw   (149,96) -- (197,96) -- (197,139) -- (149,139) -- cycle ;
\draw   (149,154) -- (197,154) -- (197,197) -- (149,197) -- cycle ;
\draw   (149,208) -- (197,208) -- (197,251) -- (149,251) -- cycle ;
\draw   (150,262) -- (198,262) -- (198,305) -- (150,305) -- cycle ;
\draw   (226,42.33) -- (274,42.33) -- (274,85.33) -- (226,85.33) -- cycle ;
\draw   (226,94.33) -- (274,94.33) -- (274,137.33) -- (226,137.33) -- cycle ;
\draw   (226,151.33) -- (274,151.33) -- (274,194.33) -- (226,194.33) -- cycle ;
\draw   (226,206.33) -- (274,206.33) -- (274,249.33) -- (226,249.33) -- cycle ;
\draw   (226,261.33) -- (274,261.33) -- (274,304.33) -- (226,304.33) -- cycle ;
\draw   (306,42.33) -- (354,42.33) -- (354,85.33) -- (306,85.33) -- cycle ;
\draw   (306,95.33) -- (354,95.33) -- (354,138.33) -- (306,138.33) -- cycle ;
\draw   (306,154.33) -- (354,154.33) -- (354,197.33) -- (306,197.33) -- cycle ;
\draw   (307,262.33) -- (355,262.33) -- (355,305.33) -- (307,305.33) -- cycle ;
\draw   (388,95) -- (436,95) -- (436,138) -- (388,138) -- cycle ;
\draw   (388,153) -- (436,153) -- (436,196) -- (388,196) -- cycle ;
\draw    (57,88.67) -- (478,89.67) ;
\draw    (61,142.67) -- (477,142.67) ;
\draw    (59,199.67) -- (475,199.67) ;
\draw    (58,254.67) -- (474,254.67) ;
\draw    (58,309.67) -- (474,309.67) ;
\draw [color={rgb, 255:red, 208; green, 2; blue, 27 }  ,draw opacity=1 ][line width=3]  [dash pattern={on 3.38pt off 3.27pt}]  (126,323.67) -- (127,149.33) ;
\draw [color={rgb, 255:red, 208; green, 2; blue, 27 }  ,draw opacity=1 ][line width=3]  [dash pattern={on 3.38pt off 3.27pt}]  (217,87.67) -- (292,87.67) ;
\draw [color={rgb, 255:red, 208; green, 2; blue, 27 }  ,draw opacity=1 ][line width=3]  [dash pattern={on 3.38pt off 3.27pt}]  (292,87.67) -- (292,43.67) ;
\draw [color={rgb, 255:red, 74; green, 144; blue, 226 }  ,draw opacity=1 ][line width=3]  [dash pattern={on 3.38pt off 3.27pt}]  (378,46.67) -- (378,204.67) ;
\draw [color={rgb, 255:red, 74; green, 144; blue, 226 }  ,draw opacity=1 ][line width=3]  [dash pattern={on 3.38pt off 3.27pt}]  (288,253.67) -- (287,326.33) ;
\draw [color={rgb, 255:red, 208; green, 2; blue, 27 }  ,draw opacity=1 ][line width=3]  [dash pattern={on 3.38pt off 3.27pt}]  (126,148.33) -- (217,148.33) ;
\draw [color={rgb, 255:red, 208; green, 2; blue, 27 }  ,draw opacity=1 ][line width=3]  [dash pattern={on 3.38pt off 3.27pt}]  (217,87.67) -- (217,149.33) ;
\draw [color={rgb, 255:red, 74; green, 144; blue, 226 }  ,draw opacity=1 ][line width=3]  [dash pattern={on 3.38pt off 3.27pt}]  (378,204.67) -- (288,253.67) ;

\draw (78,54) node [anchor=north west][inner sep=0.75pt]   [align=left] {B{\small 11}};
\draw (78,108) node [anchor=north west][inner sep=0.75pt]   [align=left] {B{\small 21}};
\draw (78,165) node [anchor=north west][inner sep=0.75pt]   [align=left] {B{\small 31}};
\draw (78,220) node [anchor=north west][inner sep=0.75pt]   [align=left] {B{\small 41}};
\draw (77,274) node [anchor=north west][inner sep=0.75pt]   [align=left] {B{\small 51}};
\draw (158,55) node [anchor=north west][inner sep=0.75pt]   [align=left] {B{\small 12}};
\draw (158,107) node [anchor=north west][inner sep=0.75pt]   [align=left] {B{\small 22}};
\draw (158,166) node [anchor=north west][inner sep=0.75pt]   [align=left] {B{\small 32}};
\draw (158,223) node [anchor=north west][inner sep=0.75pt]   [align=left] {B{\small 42}};
\draw (157,274) node [anchor=north west][inner sep=0.75pt]   [align=left] {B{\small 52}};
\draw (234,56) node [anchor=north west][inner sep=0.75pt]   [align=left] {B{\small 13}};
\draw (234,107) node [anchor=north west][inner sep=0.75pt]   [align=left] {B{\small 23}};
\draw (234,167) node [anchor=north west][inner sep=0.75pt]   [align=left] {B{\small 33}};
\draw (234,223) node [anchor=north west][inner sep=0.75pt]   [align=left] {B{\small 43}};
\draw (233,275) node [anchor=north west][inner sep=0.75pt]   [align=left] {B{\small 53}};
\draw (314,56) node [anchor=north west][inner sep=0.75pt]   [align=left] {B{\small 14}};
\draw (314,108) node [anchor=north west][inner sep=0.75pt]   [align=left] {B{\small 24}};
\draw (314,167) node [anchor=north west][inner sep=0.75pt]   [align=left] {B{\small 34}};
\draw (313,275) node [anchor=north west][inner sep=0.75pt]   [align=left] {B{\small 54}};
\draw (396,107) node [anchor=north west][inner sep=0.75pt]   [align=left] {B{\small 25}};
\draw (396,168) node [anchor=north west][inner sep=0.75pt]   [align=left] {B{\small 35}};
\draw (29,77) node [anchor=north west][inner sep=0.75pt]   [align=left] {p{\small 1}};
\draw (31,188) node [anchor=north west][inner sep=0.75pt]   [align=left] {p{\small 3}};
\draw (30,296) node [anchor=north west][inner sep=0.75pt]   [align=left] {p{\small 5}};
\draw (29,243) node [anchor=north west][inner sep=0.75pt]   [align=left] {p{\small 4}};
\draw (30,130) node [anchor=north west][inner sep=0.75pt]   [align=left] {p{\small 2}};
\draw (65,5) node [anchor=north west][inner sep=0.75pt]   [align=left] {round 1};
\draw (145,5) node [anchor=north west][inner sep=0.75pt]   [align=left] {round 2};
\draw (222,6) node [anchor=north west][inner sep=0.75pt]   [align=left] {round 3};
\draw (302,7) node [anchor=north west][inner sep=0.75pt]   [align=left] {round 4};
\draw (385,6) node [anchor=north west][inner sep=0.75pt]   [align=left] {round 5};
\draw (66,323) node [anchor=north west][inner sep=0.75pt]  [color={rgb, 255:red, 208; green, 2; blue, 27 }  ,opacity=1 ] [align=left] {Sporades block 1};
\draw (223,327) node [anchor=north west][inner sep=0.75pt]  [color={rgb, 255:red, 74; green, 144; blue, 226 }  ,opacity=1 ] [align=left] {Sporades block 2};
\draw (64,344) node [anchor=north west][inner sep=0.75pt]  [color={rgb, 255:red, 208; green, 2; blue, 27 }  ,opacity=1 ] [align=left] {cmnds=[3,2,1,1,1]};
\draw (224,347) node [anchor=north west][inner sep=0.75pt]  [color={rgb, 255:red, 74; green, 144; blue, 226 }  ,opacity=1 ] [align=left] {cmnds=[4,4,4,3,3]};

\end{tikzpicture}}
    \caption{An example execution of Sporades with Mandator}
    \label{fig:sporades-mandator}
    \vspace{-5mm}
\end{figure}
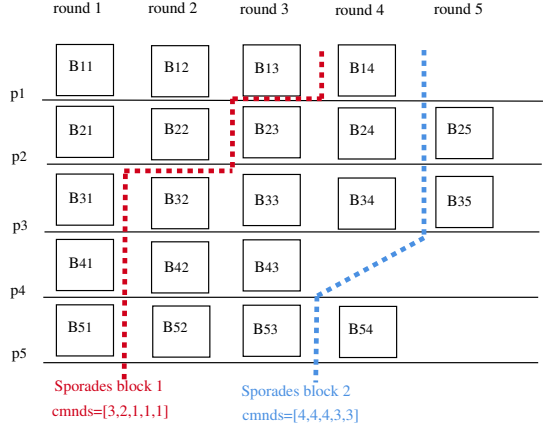

\textbf{Interaction with Mandator} Figure \ref{fig:sporades-mandator} depicts an example execution of Sporades with Mandator. There are five replicas: p\textsubscript{1}, p\textsubscript{2}, p\textsubscript{3}, p\textsubscript{4}, and p\textsubscript{5}. The 2D grid of batches B\textsubscript{i}\textsubscript{j} shows the $chains$ array stored by replica p\textsubscript{i} as part of its Mandator state. In p\textsubscript{i}'s view, replicas p\textsubscript{1}, p\textsubscript{2}, p\textsubscript{3}, p\textsubscript{4}, and p\textsubscript{5} have 4, 5, 5, 3 and 4 number of Mandator batches, each of which have received $n-f = 3$ number of <Mandator-vote>. Sporades block 1 with rank (0,1) contains [3, 2, 1, 1, 1] as the $cmnds$ field. When replica p\textsubscript{i} commits Sporades block 1,  p\textsubscript{i} commits (1) B\textsubscript{1}\textsubscript{1}, B\textsubscript{1}\textsubscript{2}, and B\textsubscript{1}\textsubscript{3} from p\textsubscript{1}'s Mandator batches, (2) B\textsubscript{2}\textsubscript{1} and B\textsubscript{2}\textsubscript{2}  from p\textsubscript{2}'s Mandator batches, (3) B\textsubscript{3}\textsubscript{1} from p\textsubscript{3}'s Mandator batches, (4) B\textsubscript{4}\textsubscript{1} from p\textsubscript{4}'s Mandator batches, and (5) B\textsubscript{5}\textsubscript{1} from p\textsubscript{5}'s Mandator batches. When Sporades block 2 from view 0 and round 2 is committed, p\textsubscript{i} commits (1) B\textsubscript{1}\textsubscript{4} from p\textsubscript{1}'s Mandator batches, (2) B\textsubscript{2}\textsubscript{3} and B\textsubscript{2}\textsubscript{4}  from p\textsubscript{2}'s Mandator batches, (3) B\textsubscript{3}\textsubscript{2}, B\textsubscript{3}\textsubscript{3}, and B\textsubscript{3}\textsubscript{4} from p\textsubscript{3}'s Mandator batches, (4) B\textsubscript{4}\textsubscript{2} and B\textsubscript{4}\textsubscript{3} from p\textsubscript{4}'s Mandator batches, and (5) B\textsubscript{5}\textsubscript{2} and B\textsubscript{5}\textsubscript{3} from p\textsubscript{5}'s Mandator batches.

\begin{algorithm}[h!] \small
\caption{Sporades Synchronous Protocol for process $p_i$, $i\in 0..n-1$} \label{algo:fallbacksteadyalgorithm}
\begin{algorithmic}[1]
\STATE for each view \textbf{v} a designated synchronous leader exists, and is denoted by \textbf{L{\textsubscript{v}}}\label{S:line:1}

\STATE \textbf{Local State:}
\begin{ALC@g}
\STATE  v\textsubscript{cur}, the current view number
\STATE  r\textsubscript{cur}, the current round number
\STATE  block\textsubscript{high}, the block with the highest round number in which the node received a \textbf{propose} message. initialized to the genesis block
\STATE  block\textsubscript{commit}, the last committed block. All blocks in the path from the block\textsubscript{commit} to the genesis block are considered committed blocks
\STATE  isAsync: boolean variable which keeps track of the current algorithm, initialized to \FALSE
\STATE  B\textsubscript{fall}[] initialized to all null, keeps the height 2 asynchronous blocks received from each node in the most recent view change
\newline
\end{ALC@g}

\STATE \textbf{Upon} receiving $n-f$ $<$vote, v, r, B$>$ with the same (v, r) \AND (v, r) $\geq$ (v\textsubscript{cur}, r\textsubscript{cur}) \AND isAsync=\FALSE \label{S:line:9}
\begin{ALC@g}
\STATE update block\textsubscript{high} to be the block in the vote messages with the highest rank
\IF{$n-f$ received vote messages have the same block\textsubscript{high} \AND the rank of this received block\textsubscript{high} is (v,r)} \label{S:line:11}
\STATE set block\textsubscript{commit} to block\textsubscript{high} \label{S:line:12}
\ENDIF 
\STATE  set v\textsubscript{cur}, r\textsubscript{cur} to v, r
\IF{ L\textsubscript{v\textsubscript{cur}} == self}
\STATE set cmnds to getClientRequests() \hspace{1mm} // from Mandator
\STATE form a new block B=(cmnds, v\textsubscript{cur}, r\textsubscript{cur}+1, block\textsubscript{high})
\STATE broadcast $<$propose, B, block\textsubscript{commit}$>$ \label{S:line:18}
\ENDIF
\newline
\end{ALC@g}

\STATE \textbf{Upon} receiving $<$propose, B, block\textsubscript{c}$>$ such that B.rank $>$ (v\textsubscript{cur}, r\textsubscript{cur}) \AND isAsync=\FALSE \label{S:line:20}
\begin{ALC@g}
\STATE  cancel Timeout
\STATE  set v\textsubscript{cur}, r\textsubscript{cur} to B.v, B.r \label{S:line:22}
\STATE  set block\textsubscript{high} to B \label{S:line:23}
\STATE  set block\textsubscript{commit} to \label{S:line:24} block\textsubscript{c}
\STATE  send $<$vote, v\textsubscript{cur}, r\textsubscript{cur}, block\textsubscript{high}$>$ to L\textsubscript{v\textsubscript{cur}} \label{S:line:25}
\STATE  set Timeout \label{S:line:26}
\newline
\end{ALC@g}

\STATE \textbf{Upon} timeout \label{S:line:27}
\begin{ALC@g}
\STATE  broadcast $<$timeout, v\textsubscript{cur}, r\textsubscript{cur}, block\textsubscript{high}$>$ \label{S:line:28}
\end{ALC@g}

\end{algorithmic}
\end{algorithm}
\vspace{-2mm}

\begin{algorithm}[h!] \small
\caption{Sporades Asynchronous Protocol for process $p_i$, $i\in 0..n-1$} \label{algo:fallbackalgorithm}
\begin{algorithmic}[1]

\STATE \textbf{Upon} receiving $n-f$ $<$timeout, v, r, b$>$ messages for the same v such that v $\geq$ (v\textsubscript{cur}) \AND isAsync = \FALSE\label{A:line:1}
\begin{ALC@g}
\STATE  set isAsync to \TRUE \label{A:line:2}
\STATE  update block\textsubscript{high} to be the block in the timeout messages with the highest rank
\STATE  set v\textsubscript{cur}, r\textsubscript{cur} to v, max(r\textsubscript{cur}, block\textsubscript{high}.r) 
\STATE  set cmnds to getClientRequests()  \hspace{1mm} // from Mandator
\STATE  form a new height 1 asynchronous block \\ \hspace{1cm} B\textsubscript{f1}=(cmnds, v\textsubscript{cur}, r\textsubscript{cur}+1, block\textsubscript{high}, 1) \label{A:line:6}
\STATE  broadcast $<$propose-async, B\textsubscript{f1},p\textsubscript{i}, 1$>$ \label{A:line:7}
\newline
\end{ALC@g}

\STATE \textbf{Upon} receiving $<$propose-async, B, $p_j$, h$>$ from $p_j$ \AND B.v  == v\textsubscript{cur} \AND  isAsync == \TRUE \label{A:line:8}
\begin{ALC@g}
\IF{rank(B) $>$ (v\textsubscript{cur}, r\textsubscript{cur})} \label{A:line:9}
\STATE send $<$vote-async, B, h$>$ to $p_j$ \label{A:line:10}
\IF{h == 2} 
\STATE set B\textsubscript{fall}[$p_j$] to B  
\ENDIF
\ENDIF
\newline
\end{ALC@g}

\STATE \textbf{Upon} receiving $n-f$ $<$vote-async, B, h$>$ \AND isAsync = \TRUE~\AND B.v == v\textsubscript{cur} \label{A:line:15}
\begin{ALC@g}
\IF{h == 1} \label{A:line:16}
\STATE set cmnds to getClientRequests()  \hspace{1mm} // from Mandator
\STATE form a new height 2 asynchronous block \\ \hspace{1cm} B\textsubscript{f2}=(cmnds, v\textsubscript{cur}, B.r+1, B, 2) \label{A:line:18}
\STATE broadcast $<$propose-async, B\textsubscript{f2}, p\textsubscript{i}, 2$>$ \label{A:line:19}
\ENDIF
\IF{h == 2}
\STATE broadcast $<$asynchronous-complete, B, v\textsubscript{cur}, p\textsubscript{i}$>$ \label{A:line:22}
\ENDIF
\newline
\end{ALC@g}

\STATE \textbf{Upon} receiving $n-f$ $<$asynchronous-complete, B, v, $p_j$$>$ \AND isAsync = \TRUE \hspace{1mm}  \AND v == v\textsubscript{cur} \label{A:line:24}
\begin{ALC@g}
\STATE  leader = invoke common-coin-flip(v\textsubscript{cur}) \label{A:line:25}
\IF{height 2 block by leader exists in the first $n-f$ asynchronous-complete messages received} \label{A:line:26}
\STATE set block\textsubscript{high}, block\textsubscript{commit}  to height 2 block from leader \label{A:line:27}
\STATE set v\textsubscript{cur}, r\textsubscript{cur} to rank(block\textsubscript{high}) \label{A:line:28}

\ELSIF{B\textsubscript{f}[leader] != null} \label{A:line:29}
\STATE set block\textsubscript{high} to  B\textsubscript{f}[l] 
\STATE set v\textsubscript{cur}, r\textsubscript{cur} to rank(block\textsubscript{high}) \label{A:line:31}
\ENDIF

\STATE set v\textsubscript{cur} to v\textsubscript{cur}+1 \label{A:line:33}
\STATE set isAsync to \FALSE \label{A:line:34}
\STATE send $<$vote, v\textsubscript{cur}, r\textsubscript{cur}, block\textsubscript{high}$>$ to L\textsubscript{v\textsubscript{cur}} \label{A:line:35}
\STATE set timeout \label{A:line:36}
\end{ALC@g}

\end{algorithmic}
\end{algorithm}
\vspace{-2mm}

\section{Implementation}\label{sec:impl}


We implemented Sporades and Mandator using Golang version 1.18.3. We first implemented Mandator and implemented Sporades and Paxos on top of Mandator, which we refer to as Mandator-Sporades and Mandator-Paxos, respectively. We used standard Golang TCP library to achieve reliable point-to-point links. We used Protobuf encoding~\cite{protobuf} to marshal and unmarshal messages. 


In our implementation, we used a classic design principle in systems design; separation of control plane from the data plane. We optimized the reliable request dissemination of Mandator by delegating the task of client request dissemination to a separate \textbf{child process} that runs in the same replica machine: each replica has one or more child processes that are assigned to it. Child processes are stateless, and are concerned only about reliably broadcasting client request batches. With child processes in place, the typical execution of Mandator is as follows. The client sends a batch of requests to a child process, the child process collects one or more such client request batches, forms a child-batch, and sends it to a majority of child processes in other replica machines. Each child process, upon receiving a new child-batch, will send an acknowledgement to the originator, and will also send the received child batch to the main replica process. The sending child process, upon receiving a majority of acknowledgements, will send a child-batch-confirm message to the replica. The replica uses the confirmed child batch identifiers in the Mandator batches.


We also employed \textbf{selective-broadcast}, a novel memory-efficient pull-based broadcast primitive in the Mandator implementation. Namely, each Mandator replica keeps track of the last received <Mandator-vote> for each replica and sends the new Mandator-batches only to the replicas that are most up-to-date, with the restriction that there is at least a majority of receivers. Replicas that do not receive the new batch use a pull-based mechanism to request missing batches. Selective broadcast helps to avoid memory overflow resulting from increasing queue lengths in the face of network asynchrony.
\section{Evaluation}\label{sec:eval}


The goal of our evaluation is to answer the following four questions: (1) How efficient is Mandator in disseminating client requests in the wide-area? (2) How robust is Mandator-Sporades in the face of replica failures? (3) How robust is Mandator-Sporades in the face of targeted DDoS attacks? and (4) How does Mandator-Sporades scale with increasing replica count in the wide-area when deployed with Redis \cite{carlson2013redis}?

 
To showcase the performance and resilience of Mandator-Sporades, we compare Mandator-Sporades against three state-of-the-art consensus algorithms; Rabia \cite{pan2021rabia}, EPaxos (with 3\% conflict rate) \cite{moraru2013there}, and Multi-Paxos \cite{mazieres2007paxos}. Rabia is a randomized consensus algorithm that builds on top of Ben-Or \cite{ben1983another}. EPaxos is a consensus algorithm that achieves high throughput in the cluster and low latency in the wide-area. Multi-Paxos is a popular choice in production systems due to its simplicity \cite{burrows2006chubby}.

\subsection{Setup \label{subsec:setupl}}


We conducted our experiments using c4.4xlarge (16 virtual CPUs, 30 GB memory, and up to 10 Gbps network bandwidth) and c4.2xlarge (8 virtual CPUs, 15 GB memory, and up to 10 Gbps network bandwidth) virtual machines for replicas and clients, respectively. We used Ubuntu Linux 20.04.3 LTS for all the machines. We employed the same evaluation setup as in Mencius\cite{barcelona2008mencius} and Baxos\cite{tennage2022baxos}, where each AWS location has a single replica machine and a single client machine. We experimented with three to nine replicas and three to nine clients located in nine geographically separated AWS regions in N.Virginia, Ireland, Mumbai, St.Paulo, Tokyo, Oregon, Ohio, Singapore and Sydney.

\subsection{Benchmarks and Workload}

We employ two benchmarks in our evaluation; (1) a resident key-value store written in Go-Lang and (2) Redis \cite{carlson2013redis}. As the resident key-value store, we use the same implementation as Rabia \cite{pan2021rabia} key-value store, which uses a Go-Lang map object to save keys and values of type string. For the Redis benchmark, we used a Go-Redis connector \cite{goredis}. Our Redis integration with Mandator-Sporades uses the same approach used in Redis-Rabia \cite{pan2021rabia}, and we use Redis native MGet and MPut commands to process a batch of requests. 

\begin{figure*}[h]
    \vspace{-8mm}
    \begin{subfigure}[b]{0.5\textwidth}
      \centering
      \includegraphics[width=\textwidth]{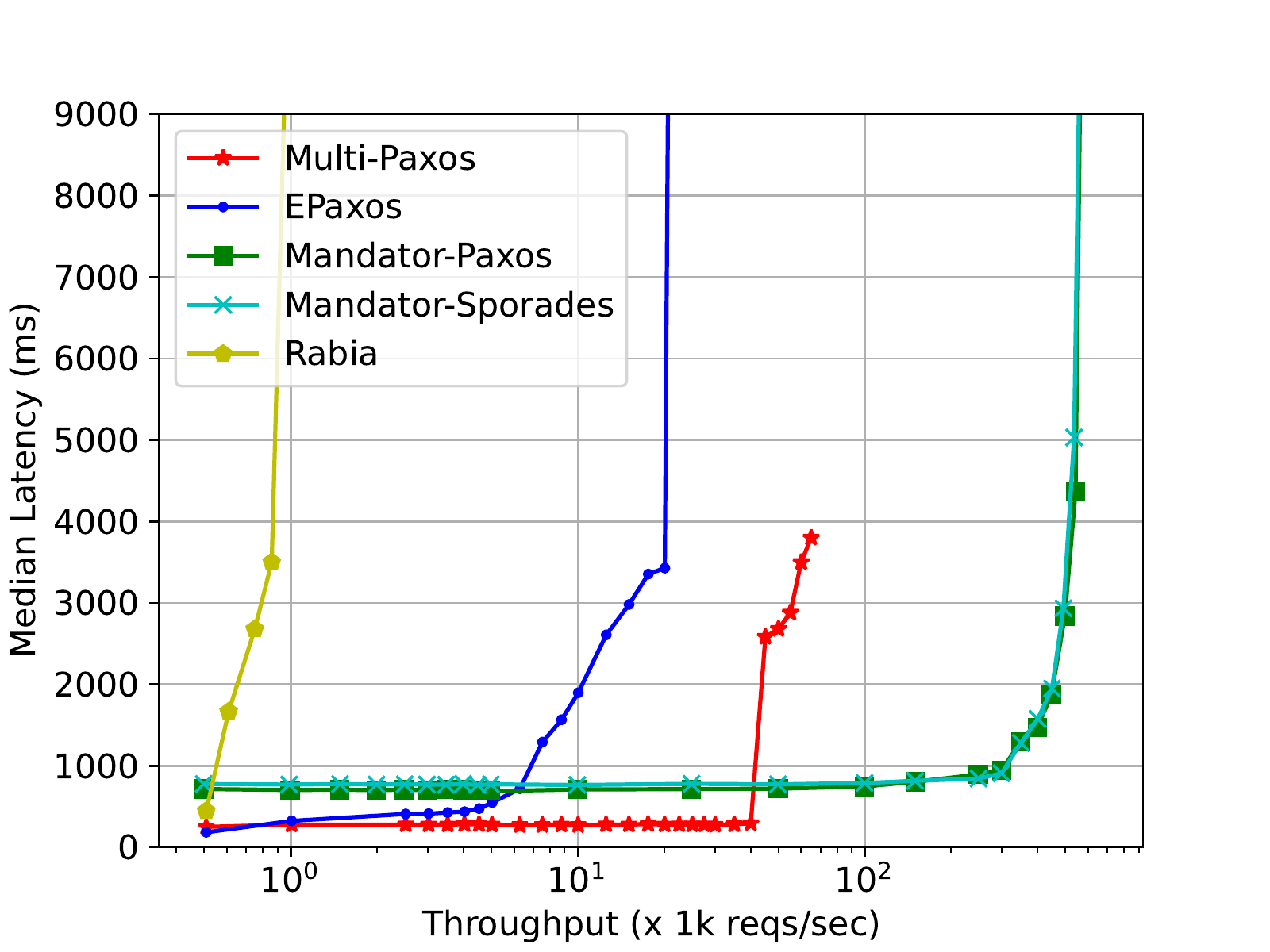}
      \caption{Median latency}
      \vspace{-1mm}
    \end{subfigure}
    \begin{subfigure}[b]{0.5\textwidth}
      \centering
      \includegraphics[width=\textwidth]{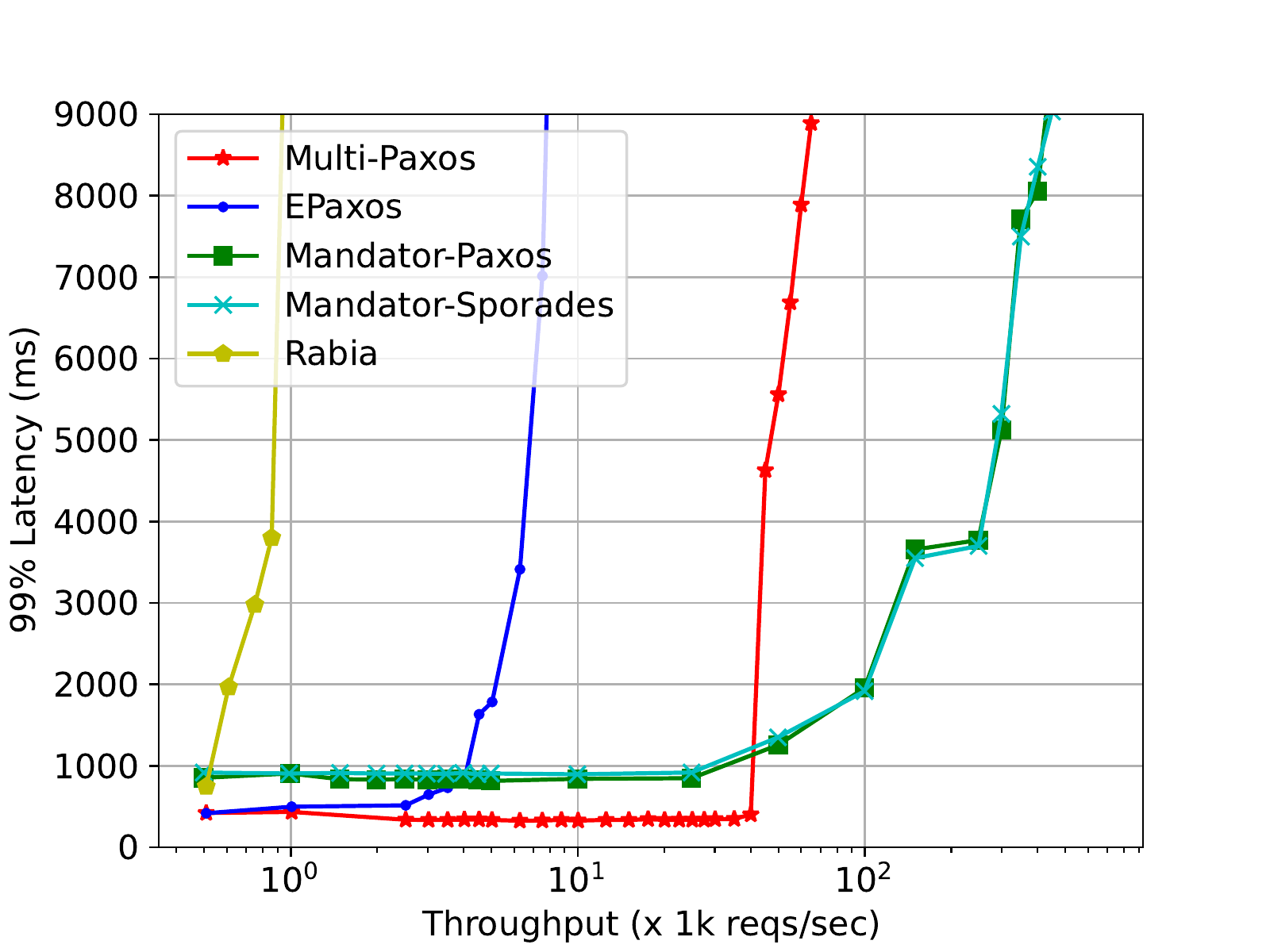}
      \caption{Tail latency}
      \vspace{-1mm}
    \end{subfigure} 
	\caption{Best-case wide-area performance of Mandator with 5 replicas and 5 clients. Note that the x-axis is in the log-scale}
	\label{fig:best_case_graph}
	\vspace{-0mm}

\end{figure*}


We use open loop model \cite{schroeder2006open} Poisson arrival of client requests, due to its wide adoption in client request generation tools \cite{tennage2022baxos, tollman2021epaxos}. All four consensus algorithms use batching on the replica side and the client side. We set the client-side batch size to 100. Interestingly, we found that different algorithms have different optimum replica batch sizes, and following the previous studies \cite{moraru2013there, pan2021rabia} we use 1000, 5000, and 300 as the replica side batch sizes for EPaxos, Multi-Paxos, and Rabia, respectively. For Sporades, we used 2000 as the replica-side batch size. Furthermore, we set the maximum batch time to 5ms which is the standard batch time used in EPaxos \cite{moraru2013there}. We do not use pipe-lining in our evaluation, which is an orthogonal approach to improve wide-area performance.  We set the size of a client request to 16B (8B key length and 8B value length). We chose these request sizes after observing the request sizes used in deployed systems; e.g. in Facebook’s production TAO system \cite{bronson2013tao}, 50\% of requests have a value field smaller than 16B.


In our experiments, all client requests are totally ordered and executed. When a replica receives a request, it uses consensus to totally order it, and upon committing and executing, sends a response back to the client. After each experiment, we use the replica commit log to verify that each replica learns the same history of requests. In all algorithms, we measure the execution latency\footnote{commit latency: latency for total ordering}  \footnote{execution latency: commit latency + latency for executing the command and the causal history} \footnote{For Multi-Paxos, Sporades and Rabia the commit latency and the execution latency differ by a negligible amount whereas for EPaxos the execution latency is significantly higher than the commit latency when the conflict rate of EPaxos is non-zero} in the client side starting from when a new request is first sent by a client until the client receives the corresponding response. Following the evaluation approach of the revised version of EPaxos \cite{tollman2021epaxos}, we opted to measure the execution latency instead of the commit latency, because the commit latency does not represent all the application use-cases; if a client depends on the response returned by an operation, then the client has to wait for the execution to be completed \cite{hunt2010zookeeper} \cite{carlson2013redis}. Each experiment was run for 60 seconds and was repeated 5 times.

\subsection{Mandator Performance \label{subsec:bestcase}}

This experiment evaluates the performance of Mandator in the wide-area deployment. We deployed the consensus algorithms in five AWS regions located in N.Virginia, Ireland, Mumbai, St.Paulo, and Tokyo. We configured the clients to send requests with different arrival rates in a Poisson open loop. We measure the throughput, median and 99\% latency for different arrival rates. Figure \ref{fig:best_case_graph} depicts the experiment results.

\textbf{Rabia}: We observe that the throughput of Rabia saturates at 500 tx/sec in less than 500ms median latency. This is a significantly different result compared to the Rabia performance evaluated inside a local area network (LAN) (inside a LAN Rabia delivers 300k tx/sec throughput \cite{pan2021rabia}). We identified that this mismatch between the local-area and wide-area performance is a direct consequence of Rabia design; Rabia is designed to achieve high performance when each replica receives the same client request approximately at the same time. In a LAN deployment this timing assumption holds, but not in the WAN. Since the performance of Rabia under best-case deployment provides low throughput, we do not use Rabia in our following experiments in sections \ref{subsec:crash} \ref{subsec:ddos} and \ref{subsec:redis}.
 
\textbf{EPaxos and Multi-Paxos}: We observe that the throughput of Epaxos saturates at 6,500 tx/sec in less than 720ms median latency. We also observe that Multi-Paxos delivers a throughput of 40,000 tx/sec in less than 295ms median latency, which is significantly higher than that of EPaxos. These behaviours conform to the observations made in the revised evaluation of EPaxos \cite{tollman2021epaxos}. We explain these behaviours as follows.

In our experiments, we measured the execution latency (we used the $-exec$ and $-dreply$ flags of EPaxos). When measured for the execution latency, EPaxos fails at achieving high throughput due to two main reasons: (1) infinitely growing dependency chains and (2) high execution delay for dependent requests. When the conflict rate of EPaxos is non-zero, it has been shown that the execution latency of EPaxos can be more than 2x the commit latency \cite{tollman2021epaxos}. Hence Epaxos delivers low throughput under a specified median latency upper bound. Moreover, for both Multi-Paxos and EPaxos, we observed that the throughput is bottlenecked by the speed of consensus:  since EPaxos and Multi-Paxos have to carry a batch of client requests inside the consensus messages, the maximum attainable throughput is bottlenecked by the speed of consensus, thus delivering sub-optimal results.

\textbf{Mandator-Paxos and Mandator-Sporades}: We make two observations in Mandator-Paxos and Mandator-Sporades: (1) both Mandator-Paxos and Mandator-Sporades deliver the same throughput and latency and (2) Mandator-Paxos and Mandator-Sporades achieve 300,000 tx/sec throughput in less than 900ms median latency, and outperform Rabia, Epaxos and Multi Paxos by 650\% in throughput, at a 200\% higher latency.

First, Mandator-Paxos and Mandator-Sporades achieve the same best case performance because the message complexity of Paxos and Sporades are similar in synchronous execution. Second, the 650\% increase of throughput in Mandator-Paxos and Mandator-Sporades is a direct consequence of the Mandator design: Mandator uses an efficient method of batching multiple client requests in a single Mandator batch. Consequently, when a single consensus block gets committed, a larger batch of client requests gets committed. Mandator-based protocols disseminate the client requests without waiting for consensus acknowledgment, hence utilizing the network link bandwidth more efficiently. Stated differently, the performance of Mandator is bottlenecked by the speed of the network. In contrast, Multi-Paxos, EPaxos, and Rabia employ batching inside the critical path of consensus, such that the throughput is bottlenecked by the speed of consensus. This validates our hypothesis that separating the request dissemination from the total-ordering improves throughput.

\textbf{Latency overhead of Mandator}: The median and tail latency of Mandator-Paxos and Mandator-Sporades are 200\% higher than Multi-Paxos for throughput values less than 40k tx/sec. Mandator-Paxos and Mandator-Sporades consume 10 message hops for each client request when deployed with a single child process, and this results in high latency. This latency cost is acceptable for applications such as strongly consistent distributed file sharing, where throughput is more critical than latency. In contrast, if the latency is more critical, then we propose to use Sporades without Mandator. When deployed without Mandator, Sporades provides the same latency and throughput behaviors as Multi-Paxos, because the asymptotic message complexities and the communication patterns are similar between synchronous mode of Sporades and Multi-Paxos. A more practical approach would be to dynamically enable Mandator depending on the application SLO latency guarantees and client arrival rate.

\subsection{Crash failures\label{subsec:crash}}


In this experiment, we evaluate the impact of replica failures on throughput. In Mandator-Paxos and Mandator-Sporades, we crash the leader replica and in EPaxos we crash a random replica (since EPaxos is leaderless). Figure \ref{fig:crash_fault_performance} shows the throughput with respect to time.

\begin{figure}[t]
  \vspace{-5mm}
  \includegraphics[width=\linewidth]{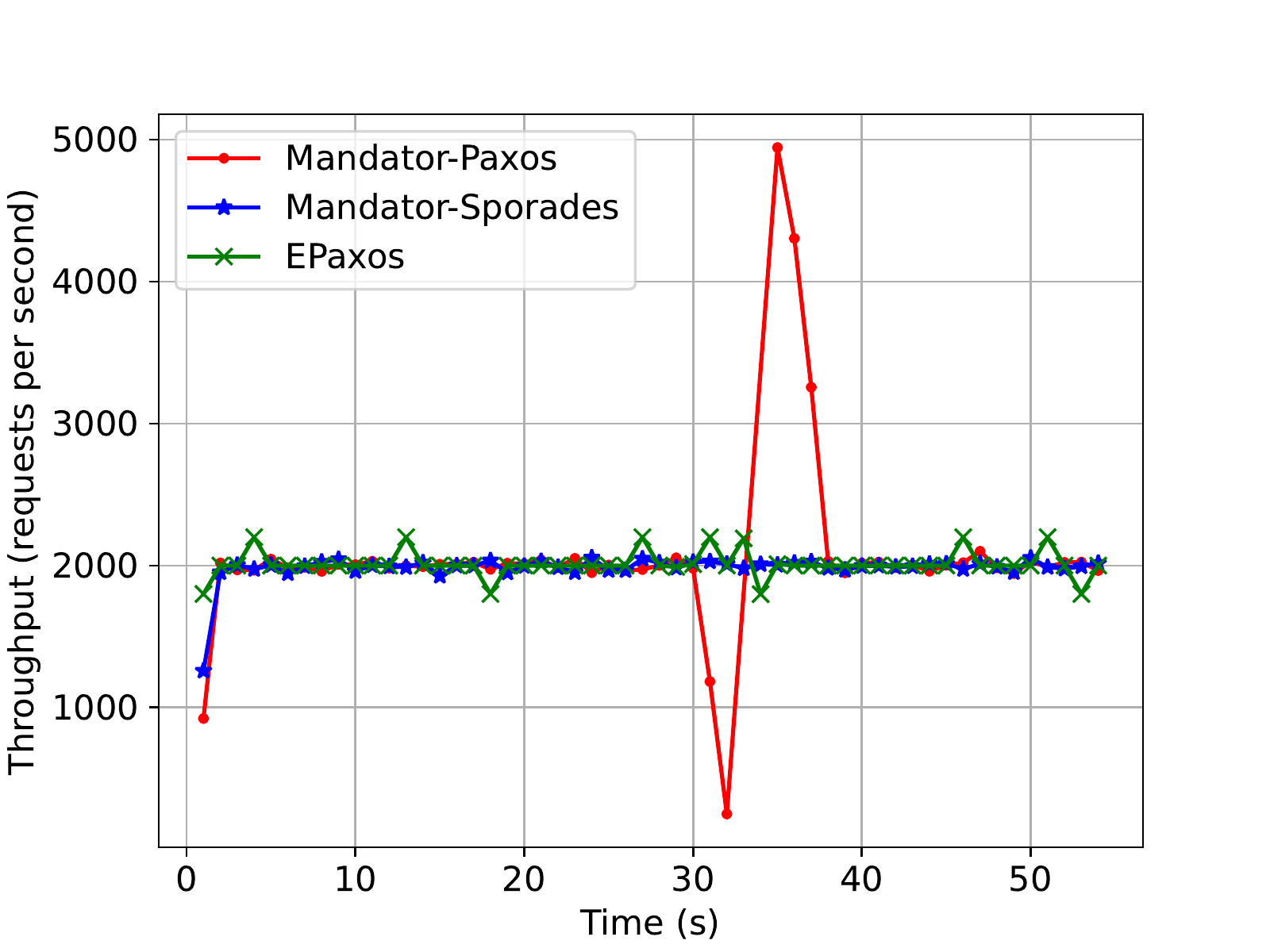}
  \caption{Crash-Fault Performance with 3 replicas and 2 clients. At t=30s the leader replica in Mandator-Paxos and Mandator-Sporades and a random replica in EPaxos are crashed}
  \vspace{-5mm}
  \label{fig:crash_fault_performance}
\end{figure}

We observe that both EPaxos and Mandator-Sporades are not affected by the node failure, whereas Mandator-Paxos has transient low and high spikes in throughput. EPaxos is resilient to node failures given that there exists no leader replica. Mandator-Sporades is resilient to leader node failures because of the asynchronous path of Sporades: in the face of leader failures, Sporades falls back to an asynchronous path which commits transactions. Hence when a new leader is elected Sporades keeps committing new requests. In contrast, Mandator-Paxos (and also Multi-Paxos) fails at committing new requests during the view change, hence the throughput approaches zero just after the leader failure. However, during this period when a new leader is elected, Mandator continues to reliably broadcast client requests, and as a result, when a new Multi-Paxos leader is elected, requests are committed with high throughput (with 5,000 tx/sec as the maximum in this experiment). This is the reason for observing an upward spike for Mandator-Paxos after the leader election.

\subsection{Performance under DDoS attacks \label{subsec:ddos}}


In this experiment, we focus on the impact of targeted DDoS attacks on throughput and latency. We use a more generalized version of \textit{delayed-view-change attack}; a DDoS attack previously presented in \cite{tennage2022baxos, spiegelman2019ace}. In a delayed-view-change-attack, a DDoS attack is simulated by continuously attacking the leader node dynamically upon each view change. However, a delayed-view-change-attack does not consider the overhead on the attacker, given that performing a dynamic network traffic analysis to find the new leader node is a time-consuming task. In our work, we modify this attack such that the attacker randomly chooses a minority of the replicas and attacks them. We believe that our approach is more practical, given that the attacker does not have to do a wide-area network traffic analysis to find the leader.

Figure \ref{fig:attack__performance} shows the throughput and median latency under DDoS attacks. Since this attack measures the performance under worst-case network scenarios, in the following we discuss the throughput values under 5s median latency upper bound (in contrast, we discussed throughput under 1s median latency in the best case evaluation in subsection \ref{subsec:bestcase}). We first observe that the attack has the highest impact on EPaxos by having a maximum throughput of 7200 tx/sec (under 5s median latency). This is due to the high delayed-execution latency resulting from the expensive instance revocation path of the EPaxos. Second, the throughput of Multi-Paxos is also affected by the DDoS attack by having a throughput of 45,000 tx/sec (under 5s median latency), due to liveness loss during the leader elections. Third, Mandator-Paxos delivers 250,000 tx/sec throughput (under 5s median latency) which is significantly less than its synchronous performance. However, the DDoS performance of Mandator-Paxos is still large enough due to the impact of Mandator: Mandator is not affected by the attack due to its asynchronous nature, and hence continues to reliably disseminate requests at the speed of the network. Though Paxos is affected by the DDoS attack, the impact on the average throughput of Mandator-Paxos is low, because upon a new leader election Mandator-Paxos can commit a large batch of requests that is disseminated by Mandator. Finally, Mandator-Sporades delivers 400,000 tx/sec throughput under 5s median latency. We attribute this high throughput to Sporades's liveness guarantees during network asynchrony: the asynchronous mode of Sporades continues to commit client requests during network asynchrony. Hence this experiment confirms our design hypothesis that Mandator-Sporades is resilient to DDoS attacks due to its asynchronous liveness guarantees.

\begin{figure}[t]
  \vspace{-5mm}
  \includegraphics[width=\linewidth]{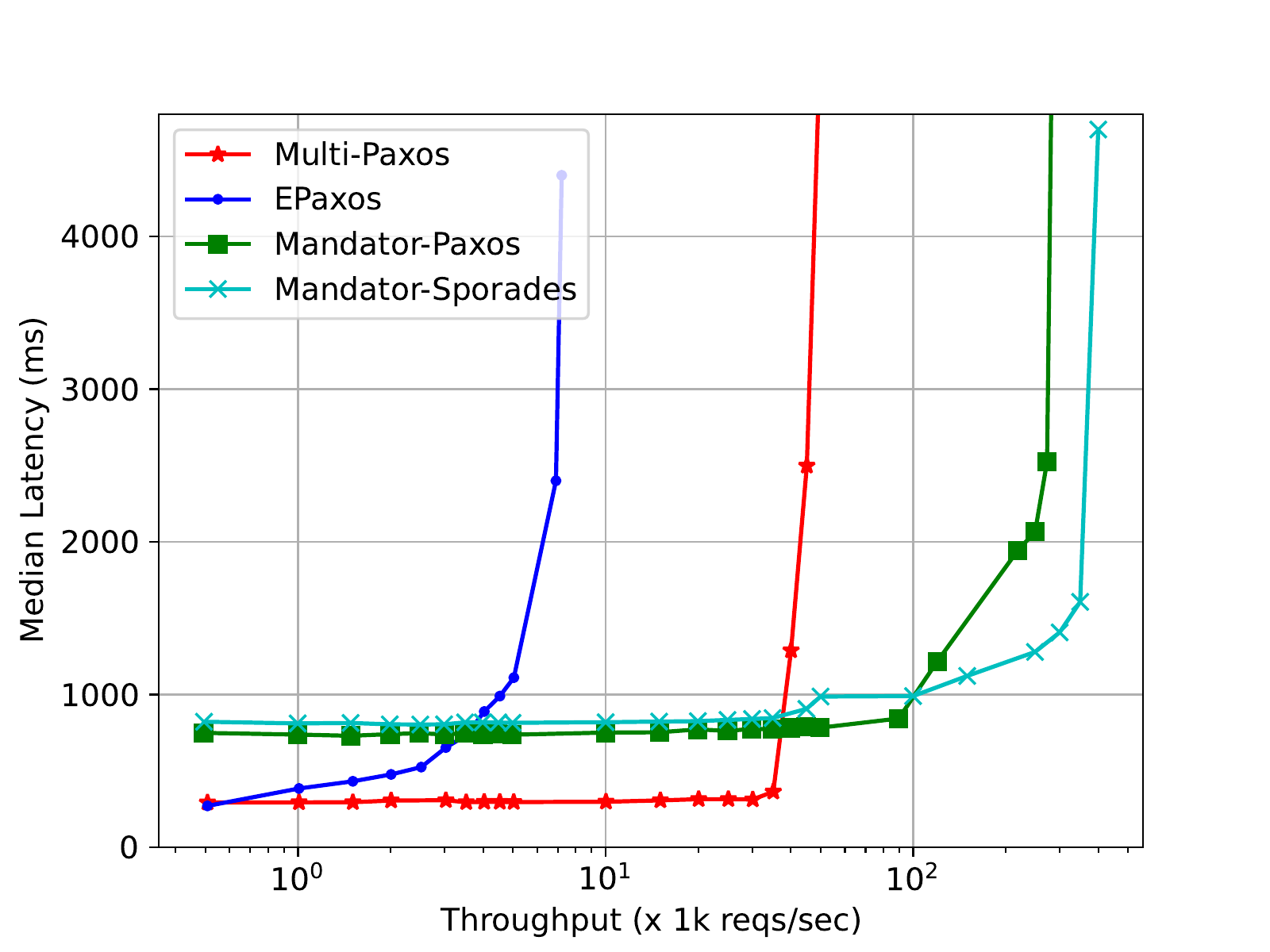}
  \caption{Performance under DDoS attack with 5 replicas. Up to a minority (=2) of nodes are attacked at a given time, dynamically. Note that the x-axis is in log-scale.}
  \vspace{-5mm}
  \label{fig:attack__performance}
\end{figure}

\subsection{Scalability with Redis\label{subsec:redis}}


In this experiment, we aim to quantify the scalability of Mandator-Sporades with Redis integration. Redis is an in-memory data store that supports data structures and operations on them, such as hash maps, lists, and sets. Since the real-world deployments have strict service level objectives for the latency, we fix the maximum allowed median latency to be 1.5s. We evaluated the scalability of Mandator-Sporades by running it with an ensemble of three, five, seven and nine replicas. It should be noted that, unlike blockchain algorithms where consensus algorithms are measured for up to a hundred nodes, crash fault tolerant protocols are designed to scale up to 9 nodes in practice \cite{barcelona2008mencius, kogias2020hovercraft}. Figure \ref{fig:scalability_performance} shows the scalability results.

We observe that Mandator-Sporades can scale up to 9 nodes while delivering a throughput of 150,000 tx/sec, in the 9 replica deployment. We also observe that with an increasing replica count the median latency increases, and we attribute this behaviour to the increasing broadcast message overhead with the increasing number of replicas.

\begin{figure}[t]
  \vspace{-5mm}
  \includegraphics[width=\linewidth]{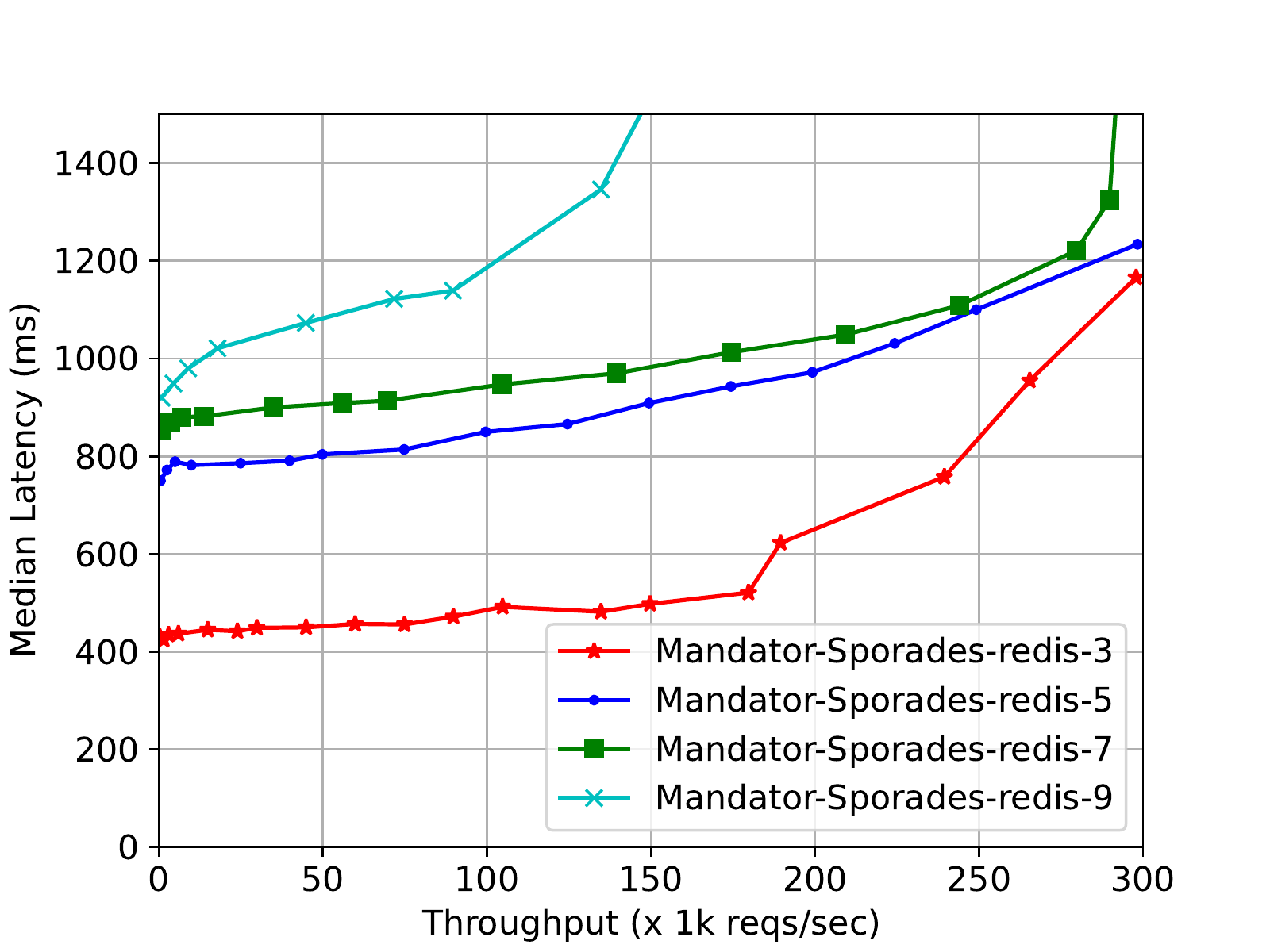}
  \caption{Scalability with respect to increasing replica count}
  \vspace{-5mm}
  \label{fig:scalability_performance}
\end{figure}

\section{Related Work}\label{sec:related}

\textbf{Partially synchronous consensus}: Leader-based partially synchronous consensus algorithms such as Multi-Paxos \cite{mazieres2007paxos}, Raft \cite{ongaro2014search}, and View Stamp Replication \cite{oki1988viewstamped} use a single leader node to totally order requests. While using the leader-based approach is simple, it suffers from leader computational and network bottlenecks, thus failing to achieve high performance. To address the single leader bottleneck, multi-leader consensus algorithms such as Baxos \cite{tennage2022baxos}, Mencius \cite{barcelona2008mencius}, EPaxos \cite{moraru2013there}, Generalized Paxos \cite{lamport2005generalized} and Multi-Coordination Paxos \cite{camargos2007multicoordinated} were proposed. At the core of the multi-leader design, there are two design principles: (1) statically partition the replicated log space and (2) exploit request commutativity to parallelize total ordering. These Multi-leader approaches suffer from low performance when concurrent requests have conflicting dependencies. All leader-based and leaderless partially synchronous protocols fail at guaranteeing liveness in the face of network asynchrony and DDoS attacks. In contrast, Sporades continues to commit requests in the presence of network asynchrony.

\textbf{Asynchronous consensus}: Ben-Or \cite{ben1983another} proposed an asynchronous binary consensus algorithm using randomization. Following Ben-Or, several theoretical research \cite{abraham2018synchronous, ezhilchelvan2001randomized, friedman2005simple, rabin1983randomized} have proposed improvisations for Ben-Or that enables (1) Byzantine tolerance\cite{mostefaoui2014signature}, (2) reducing the message complexity \cite{raynal2018fault} and supporting multi-valued consensus \cite{zhang2009bounded}. However, all of these works are of theoretical interest, with no practical evaluation. To the best of our knowledge, the first practical system that employs Ben-Or as its core is Rabia \cite{pan2021rabia}. As we showed in section \ref{sec:eval}, Rabia does not provide throughput values greater than 1000 tx/sec in the wide-area. In contrast, Sporades provides 300k tx/sec throughput in the wide-area. Previous research \cite{cachin2005random, cachin2016non, miller2016honey} have proposed Byzantine asynchronous consensus algorithms, but they are not directly related to our scope.  

\textbf{Protocol switching consensus}: Turtle consensus \cite{nikolaou2015turtle, nikolaou2016moving} is the closest work to Sporades that employs a synchronous path and an asynchronous path. Turtle consensus works by statically switching between Multi-Paxos and Ben-Or. Statically switching is sub-optimal given that the asynchronous path can be triggered by a transient network delay (in a synchronous setting), such that performance will be degraded. In contrast, in Sporades, we use dynamic mode switching between the synchronous mode and the asynchronous mode, such that Sporades stays in the synchronous path whenever the network is synchronous. Several works \cite{ gelashvili2021jolteon, aublin2015next, kursawe2005optimistic} have explored the protocol switching approach with Byzantine failure tolerance for blockchain applications, but fall outside our scope.

\textbf{Efficient request dissemination} Previous research have focused on improving the request dissemination using two approaches: (1) sharding \cite{marandi2012multi} \cite{ailijiang2017wpaxos} \cite{corbett2013spanner} and (2) overlay networks \cite{marandi2010ring} \cite{rizvi2017canopus} \cite{charapko2021pigpaxos} \cite{zhao2018sdpaxos}. We find that sharding is an orthogonal approach to our design of Mandator: Mandator can be extended to support sharding by deploying Mandator in multiple shards. Overlay networks such as Ring Paxos \cite{marandi2010ring}, Canopus \cite{rizvi2017canopus}, PigPaxos \cite{charapko2021pigpaxos}, and SD-Paxos \cite{zhao2018sdpaxos} improve the performance of the consensus algorithm by using non-leader replicas to relay the requests such that the overhead imposed on the leader node is minimized. However, the performance gain of such approaches is relatively low compared to Mandator, and we find the monolithic design of existing protocols as the main cause of this problem. In contrast, we present Mandator as a stand-alone and generic consensus agnostic overlay, that can run at the speed of the network, thus delivering high performance as shown in section \ref{sec:eval}.
\section{Conclusion}\label{sec:conc}

We presented Mandator and Sporades, a novel wide-area SMR algorithm that can achieve high performance and robustness under network asynchrony. We first presented Mandator, a consensus agnostic asynchronous request dissemination layer, that efficiently disseminates client requests. Then we presented Sporades, a novel consensus algorithm that can achieve optimum one round-trip consensus in the synchronous network setting, as well as liveness in the face of network asynchrony using an asynchronous fallback path. Our evaluation shows that Mandator-Sporades delivers 300k tx/sec under 900ms in the wide-area, and out-performs Multi-Paxos by 650\%, at a modest expense of latency. We also showed that Mandator-Sporades outperforms Mandator-Paxos by 60\% in throughput, in the presence of DDoS attacks.

{\footnotesize \bibliographystyle{plain}
\bibliography{refs}}

\section{Appendix}\label{sec:other}

\subsection{Mandator Formal Proofs}\label{sec:formal-proof-mandator}
\textit{Proof of Availability}: A write($B$) succeeds when $B$ is created and sent to all the replicas, and only after receiving at least $n-f$ Mandator-votes. Since each replica saves $B$ in the $chains$ array, it is guaranteed that $B$ will persist as long as $n-f$ nodes are alive due to quorum intersection. Hence read($B$) eventually returns.

\textit{Proof of Causality}: Causality follows from the fact that each replica extends its chain of Mandator-batches, and because each replica creates the batch with round $r$ only after completing the write of the batch with round $r-1$.

\subsection{Sporades Example Execution}\label{sec:sporades-execution}

Figure~\ref{fig:consensus-regular} depicts an example execution of the synchronous mode of Sporades. The leader replica for the shown view is $p_1$. The remaining replicas are $p_2$, $p_3$, $p_4$, and $p_5$; there are only $5$ replicas in total, which means the threshold $f$ is $2$. In this execution, $p_5$ crashes immediately, whereas $p_3$ crashes during round $2$.

The rounds correspond to the round number withheld by each replica. The protocol is synchronous but wait-free; as soon as a quorum of $3$ messages are collected, the replicas can proceed to the next computational step.


$p_1$ broadcasts a <propose, $B_1$> message (line~\ref{S:line:18}) to which replicas (including itself) respond with a <vote, $B_1$> (line~\ref{S:line:25}). As soon as $3$ <vote, $B_1$> are collected (line~\ref{S:line:11}), $p_1$ commits $B_1$ (line~\ref{S:line:12}), builds a new block $B_2$ and broadcasts a <propose, $B_2$> message which includes the information that $B_1$ was committed (line~\ref{S:line:18}). Upon receiving <propose, $B_2$>, each replica commits $B_1$ (line~\ref{S:line:24}) before sending back a <vote ,$B_2$> message (line~\ref{S:line:25}). The exact same thing happens with <propose, $B_3$> and <vote, $B_3$> and $p_1$ commits $B_3$. At this moment, due to network asynchrony, $p_1$ fails at broadcasting the <propose, $B_4$> messages and all non-crashed replicas reach their timeout. All timed-out replicas broadcast a <timeout, $B_3$> message (line~\ref{S:line:27}-\ref{S:line:28}), $B_3$ being the highest block they know of. Upon receiving $n-f=3$ <timeout, $B_3$> messages, each replica enters the asynchronous mode of Sporades.


The asynchronous execution of Sporades is depicted in figure~\ref{fig:consensus-fallback}. In the asynchronous mode, each replica acts as a leader and builds its own chain. To improve readability, we only describe the events from the perspective of a single replica. The very beginning of the execution corresponds to the end of figure~\ref{fig:consensus-regular}: the receiving of <timeout, $B_3$> messages (line~\ref{A:line:1}) and the entering of the asynchronous path. We consider the execution performed by the process $p_2$. $p_2$ sets $isAsync$ to \textit{true} (line~\ref{A:line:2}). For $p_2$, the $block_{high}$ is $B_3$ so $p_2$ forms a new height $1$ asynchronous block whose parent is $B_3$ (line~\ref{A:line:6}) and broadcasts it in a <propose-async, $B_{f1}$> message (line~\ref{A:line:7}). $p_2$ receives $n-f=3$ <vote-async,$B_{f1}$> messages (line~\ref{A:line:15}), forms a new height $2$ block $B_{f2}$ whose parent is $B_{f1}$ (line~\ref{A:line:18}) and broadcasts it as a <propose-async,$B_{f2}$> message (line~\ref{A:line:19}). $p_2$ receives $n-f$ <vote-async,$B_{f2}$> messages (line~\ref{A:line:15}) and broadcasts an <asynchronous-complete, $B_{f2}$, $v_{cur}$, self.id> message (line~\ref{A:line:22}). In the meantime, $p_2$ also votes on the <propose-async> messages from $p_1$, $p_3$, $p_4$, $p_5$, and itself (line~\ref{A:line:10}), and they eventually complete each their own asynchronous steps. 


After receiving $n-f$ <asynchronous-complete, $B_{f2}$, v, p> messages, $p_2$ invokes a common-coin-flip($v$) that designates a leader (line~\ref{A:line:25}). If the leader returned from the common-coin-flip($v$) is a sender of one of the $n-f$ <asynchronous-complete, $B_{f2}$, v, p> messages received (line~\ref{A:line:26}), then that block (and all its causal history) is committed (line~\ref{A:line:27}). If the elected leader from the common-coin-flip($v$) does not fulfill that condition, but still has successfully broadcast a height $2$ block (from the perspective of non-leader replicas, they observe the height $2$ block from the elected leader), then that block becomes the new $block_{high}$ for $p_2$ (line~\ref{A:line:29}-\ref{A:line:31}). In any case, $p_2$ exits the asynchronous path by setting a timeout and sending a <vote, $v_{cur},r_{cur},block_{high}$> message to the leader of the next view (line~\ref{A:line:35}-\ref{A:line:36}). It should be noted that the leader returned by the common-coin-flip($v$) may be different from the designated synchronous path leader $L_{v}$ of view $v$.

\subsection{Sporades Formal Proofs}\label{sec:formal-proof}
\textbf{Definition}\\ \textbf{elected-asynchronous block}: We refer to an asynchronous block B\textsubscript{f} generated in view $v$ with height $2$ as an elected-asynchronous block, if the common-coin-flip($v$) returns the index of the proposer p\textsubscript{l} who generated B\textsubscript{f} in the view $v$ and if the asynchronous-complete for B\textsubscript{f} exists in the first $n-f$ asynchronous-complete messages received. An elected-asynchronous block is committed and treated as a synchronously-committed block.

\textbf{Proof of safety} 

We first show that for a given rank ($v$,$r$), there exists a unique block. In the following lemmas \ref{lem:1},\ref{lem:2},\ref{lem:3},\ref{lem:4},\ref{lem:5}, and \ref{lem:6} we consider different formations of blocks with the same rank.

\begin{lemma}\label{lem:1}
    Let $B$, $\tilde{B}$ be two synchronous blocks with rank $(v,r)$. Then $B$ and $\tilde{B}$ are the same.
\end{lemma}
\begin{proof}
    Assume by contradiction $B$ and $\tilde{B}$ to be different. Then, according to algorithm~\ref{algo:fallbacksteadyalgorithm} line~\ref{S:line:9}, both $B$ and $\tilde{B}$ received $n-f>\frac{n}{2}$ votes. Since no replica can equivocate, i.e. no replica sends a <vote> for two different blocks with the same rank $(v,r)$, this is a contradiction. Thus $B = \tilde{B}$.
\end{proof}

\begin{lemma}\label{lem:2}
    Let $B$, $\tilde{B}$ be two elected-asynchronous blocks with rank $(v,r)$ and height $2$. Then $B$ and $\tilde{B}$ are the same.
\end{lemma}
\begin{proof}
    Assume by contradiction $B$ and $\tilde{B}$ to be different. Then, according to algorithm~\ref{algo:fallbackalgorithm} line~\ref{A:line:25}-\ref{A:line:27}, both leaders who sent $B$ and $\tilde{B}$ in an <asynchronous-complete> message were elected with the same common-coin-flip($v$). Since no replica can equivocate, i.e. no replica sends <asynchronous-complete> message for two different blocks with the same rank $(v,r)$, and because the common-coin-flip($v$) returns a unique leader for each $v$, this is a contradiction. Thus $B = \tilde{B}$.
\end{proof}

\begin{lemma}\label{lem:3}
    Let $B$, $\tilde{B}$ be two asynchronous blocks with rank $(v,r)$ and height $1$, such that both blocks are parents of an elected-asynchronous block of the same view $v$. Then $B$ and $\tilde{B}$ are the same.
\end{lemma}
\begin{proof}
    Assume by contradiction $B$ and $\tilde{B}$ to be different. Then both $B$ and $\tilde{B}$ can have a distinct child height $2$ elected-asynchronous block with rank $(v,r+1)$ (see algorithm~\ref{algo:fallbackalgorithm} line~\ref{A:line:18}). According to lemma~\ref{lem:2}, this is a contradiction. Thus $B = \tilde{B}$.
\end{proof}

\begin{lemma}\label{lem:4}
    Let $B$ be a synchronous block which receives $n-f$ <vote>s. Then there cannot exist a height $1$ asynchronous block $\tilde{B}$ that is a parent of a height $2$ elected-asynchronous block where $B$ and $\tilde{B}$ have rank $(v,r)$.
\end{lemma}
\begin{proof}
    Assume by way of contradiction that $\tilde{B}$ exists. Because $B$ received $n-f$ votes, at least $n-f$ replicas saw $B$ before $\tilde{B}$ (see algorithm~\ref{algo:fallbacksteadyalgorithm} line~\ref{S:line:20}). $\tilde{B}$ has received $n-f$ <vote-async> (see algorithm~\ref{algo:fallbackalgorithm} line~\ref{A:line:15}) from replicas who could not have seen $B$ before (see algorithm~\ref{algo:fallbackalgorithm} line~\ref{A:line:9}). Because $n-f>\frac{n}{2}$, this is a contradiction. Hence $\tilde{B}$ does not exist.
\end{proof}

\begin{lemma}\label{lem:5}
    Let $B$ be a synchronous block which receives $n-f$ votes with rank $(v,r)$. Then there cannot exist an elected-asynchronous block $\tilde{B}$ of height $2$ with rank $(v,r)$.
\end{lemma}
\begin{proof}
    Assume by way of contradiction that $\tilde{B}$ exists. Because $B$ received $n-f$ votes, at least $n-f$ replicas saw $B$ before $\tilde{B}$ (see algorithm~\ref{algo:fallbacksteadyalgorithm} line~\ref{S:line:20}). $\tilde{B}$ has received $n-f$ <vote-async> (see algorithm~\ref{algo:fallbackalgorithm} line~\ref{A:line:15}) from replicas who could not have seen $B$ before (see algorithm~\ref{algo:fallbackalgorithm} line~\ref{A:line:9}). Because $n-f>\frac{n}{2}$, this is a contradiction. Hence $\tilde{B}$ does not exist.
\end{proof}

\begin{lemma}\label{lem:6}
    Let $B$ be a height $1$ asynchronous block that is the parent of height 2 elected-asynchronous block in the same view. Then there cannot exist a height $2$ elected-asynchronous block $\tilde{B}$ with rank $(v,r)$.
\end{lemma}
\begin{proof}
    Assume by way of contradiction that $\tilde{B}$ exists. The height $1$ parent block of $\tilde{B}$ had rank $(v,r-1)$ (see algorithm~\ref{algo:fallbackalgorithm} line~\ref{A:line:18}) and was created after receiving $n-f$ timeout messages with rank $(v,r-2)$ (see algorithm~\ref{algo:fallbackalgorithm} line~\ref{A:line:6}). On the other hand, $B$ was created after receiving $n-f$ timeout messages with rank $(v,r-1)$. Because $n-f>\frac{n}{2}$, this is a contradiction. Hence $\tilde{B}$ does not exist.
\end{proof}

\begin{theorem}\label{th:1}
    Let $B$ and $\tilde{B}$ be two blocks with rank $(v,r)$. Each of $B$ and $\tilde{B}$ can be of type: (1) synchronous block which collects at least $n-f$ votes or (2) elected-asynchronous block of height $2$ or (3) height $1$ asynchronous block which is a parent of an elected-asynchronous block. Then $\tilde{B}$ and $B$ are the same.
\end{theorem}
\begin{proof}
    This holds directly from Lemma~\ref{lem:1},~\ref{lem:2},~\ref{lem:3},~\ref{lem:4},~\ref{lem:5} and~\ref{lem:6}.
\end{proof}

\begin{theorem}\label{th:2}
    Let $B$ and $\tilde{B}$ be two adjacent blocks, then $\tilde{B}.r = B.r+1$ and $\tilde{B}.v \geq B.v$.
\end{theorem}
\begin{proof}
    According to the algorithm, there are three instances where a new block is created.

\begin{itemize}
    \item Case 1: when $isAsync$ = false and L\textsubscript{v} receives at least $n-f$ <vote>s for round $r$ and creates a new synchronous block extending the block\textsubscript{high} (see algorithm~\ref{algo:fallbacksteadyalgorithm} line~\ref{S:line:9}). In this case, L\textsubscript{v} creates a new block with round $r+1$. Hence the adjacent blocks have monotonically increasing round numbers.
    \item Case 2: when $isAsync$ = true and upon collecting $n-f$ <timeout> messages in view $v$ (see algorithm~\ref{algo:fallbackalgorithm} line~\ref{A:line:1}). In this case, the replica selects the block\textsubscript{high} with the highest rank $(v,r)$, and extends it by proposing a height $1$ asynchronous block with round $r+1$. Hence the adjacent blocks have monotonically increasing round numbers.
    \item Case 3: when $isAsync$ = true and upon collecting $n-f$ <vote-async> messages for a height $1$ asynchronous block (see algorithm~\ref{algo:fallbackalgorithm} line~\ref{A:line:15}-\ref{A:line:16}). In this case, the replica extends the height $1$ block by proposing a height $2$ block with round $r+1$. Hence the adjacent blocks have monotonically increasing round numbers.
\end{itemize}

The view numbers are non decreasing according to the algorithm. Hence Theorem~\ref{th:2} holds.
\end{proof}

\begin{theorem}\label{th:3}
    If a synchronous block B\textsubscript{c} with rank $(v,r)$ is committed, then all future blocks with view $v$ will extend $B\textsubscript{c}$.
\end{theorem}
\begin{proof}
    We prove this by contradiction. 

    Assume there is a committed block B\textsubscript{c} with B\textsubscript{c}.r = r\textsubscript{c} (hence all the blocks in the path from the genesis block to B\textsubscript{c} are committed). Let block B\textsubscript{s} with B\textsubscript{s}.r = r\textsubscript{s} be the round r\textsubscript{s} block such that B\textsubscript{s} conflicts with B\textsubscript{c} (B\textsubscript{s} does not extend B\textsubscript{c}). Without loss of generality, assume that r\textsubscript{c} $<$ r\textsubscript{s}.

    Let block B\textsubscript{f} with B\textsubscript{f}.r = r\textsubscript{f} be the first valid block formed in a round r\textsubscript{f} such that r\textsubscript{s} $\geq$ r\textsubscript{f} $>$ r\textsubscript{c} and B\textsubscript{f} is the first block from the path from genesis block to B\textsubscript{s} that conflicts with B\textsubscript{c}; for instance B\textsubscript{f} could be B\textsubscript{s}. B\textsubscript{f} is formed in round r\textsubscript{f} upon the leader receiving $n-f$ <vote>s in round r\textsubscript{f} (see algorithm~\ref{algo:fallbacksteadyalgorithm} line~\ref{S:line:9}). Due to the minimality of B\textsubscript{f} (B\textsubscript{f} is the first block that conflicts with B\textsubscript{c}), all the block\textsubscript{high} values in the received $n-f$ votes contain either B\textsubscript{c} or a block that extends B\textsubscript{c}. Since block\textsubscript{high} with the highest round from the received set of votes extends B\textsubscript{c}, B\textsubscript{f} extends B\textsubscript{c}, thus we reach a contradiction. Hence no such B\textsubscript{f} exists. Hence all the blocks created after B\textsubscript{c} in the view $v$ extend B\textsubscript{c}.
\end{proof}

\begin{theorem}\label{th:4}
    If a synchronous block $B$ with rank $(v,r)$ is committed, an elected-asynchronous block $\tilde{B}$ of the same view $v$ will extend that block.
\end{theorem}
\begin{proof}
    We prove this by contradiction. Assume that a synchronous block $B$ is committed in view $v$ and an elected-asynchronous block $\tilde{B}$ does not extend $B$. Then, the parent height $1$ block of $\tilde{B}$, $\tilde{B\textsubscript{p}}$, also does not extend $B$.

    To form the height $1$ $\tilde{B\textsubscript{p}}$, the replica collects $n-f$ <timeout> messages (see algorithm~\ref{algo:fallbackalgorithm} line~\ref{A:line:1}), each of them containing the block\textsubscript{high}. If $B$ is committed, by theorem~\ref{th:3}, at least $n-f$ replicas should have set (and possibly sent) $B$ or a block extending $B$ as the block\textsubscript{high}. Hence by intersection of the quorums $\tilde{B\textsubscript{p}}$ extends $B$, thus we reach a contradiction.
\end{proof}

\begin{theorem}\label{th:5}
    At most one height 2 asynchronous block from one proposer can be committed in a given view change.
\end{theorem}
\begin{proof}
    Assume by way of contradiction that 2 height $2$ asynchronous blocks from two different proposers are committed in the same view. A height $2$ asynchronous block $B$ is committed in the asynchronous phase if the common-coin-flip($v$) returns the proposer of $B$ as the elected proposer (algorithm~\ref{algo:fallbackalgorithm} line~\ref{A:line:25}). Since the common-coin-flip($v$) outputs the same elected proposer across different replicas, this is a contradiction. Thus all height $2$ asynchronous blocks committed during the same view are from the same proposer.
    
    Assume now that the same proposer proposed two different height $2$ asynchronous blocks. According to the algorithm~\ref{algo:fallbackalgorithm} line~\ref{A:line:18}, and since no replica can equivocate, this is absurd.
    
    Thus at most one height $2$ asynchronous block from one proposer can be committed in a given view change.
\end{proof}

\begin{theorem}\label{th:6}
    Let $B$ be a height $2$ elected-asynchronous block that is committed, then all blocks proposed in the subsequent rounds extend $B$.
\end{theorem}
\begin{proof}
    We prove this by contradiction. Assume that height two elected-asynchronous block $B$ is committed with rank $(v,r)$ and block $\tilde{B}$ with rank ($\tilde{v}$, $\tilde{r}$) such that ($\tilde{v}$, $\tilde{r}$) $>$ $(v,r)$ is the first block in the chain starting from $B$ that does not extend $B$. $\tilde{B}$ can be formed in two occurrences: (1) $\tilde{B}$ is a synchronous block in the view $v+1$ (see algorithm~\ref{algo:fallbacksteadyalgorithm} line~\ref{S:line:9}) or (2) $\tilde{B}$ is a height $1$ asynchronous block with a view strictly greater than $v$ (see algorithm~\ref{algo:fallbackalgorithm} line~\ref{A:line:6}). (we do not consider the case where $\tilde{B}$ is a height 2 elected-asynchronous block, because this directly follows from \ref{th:1})

If $B$  is committed, then from the algorithm construction it is clear that a majority of the replicas will set $B$ as block\textsubscript{high}. This is because, to send a <asynchronous-complete> message with $B$, a replica should collect at least $n-f$ <vote-async> messages (see algorithm~\ref{algo:fallbackalgorithm} line~\ref{A:line:15}). Hence, its guaranteed that if $\tilde{B}$ is formed in view v+1 as a synchronous block, then it will observe $B$ as the block\textsubscript{high}, thus we reach a contradiction. In the second case, if $\tilde{B}$ is formed in a subsequent view, then it is guaranteed that the height $1$ block will extend $B$ by gathering from the <timeout> messages $B$ as block\textsubscript{high} or a block extending $B$ as the block\textsubscript{high} (see algorithm~\ref{algo:fallbackalgorithm} line~\ref{A:line:6}), hence we reach a contradiction.
\end{proof}

\begin{theorem}\label{th:7}
    There exists a single history of committed blocks.
\end{theorem}
\begin{proof}
    Assume by way of contradiction there are two different histories $H_1$ and $H_2$ of committed blocks. Then there is at least one block from $H_1$ that does not extend at least one block from $H_2$. This is a contradiction with theorems~\ref{th:4} and~\ref{th:6}. Hence there exists a single chain of committed blocks.
\end{proof}

\begin{theorem}\label{th:8}
    For each committed replicated log position $r$, all replicas contain the same block.
\end{theorem}
\begin{proof}
    By theorem~\ref{th:2}, the committed chain will have incrementally increasing round numbers. Hence for each round number (log position), there is a single committed entry, and by theorem~\ref{th:1}, this entry is unique. This completes the proof.
\end{proof}

\textbf{Proof of liveness}

\begin{theorem}\label{th:9}
    If at least $n-f$ replicas enter the asynchronous phase of view $v$ by setting $isAsync$ to true, then eventually they all exit the asynchronous phase and set $isAsync$ to false.
\end{theorem}
\begin{proof}
    If $n-f$ replicas enter the asynchronous path, then eventually all replicas (except for failed replicas) will enter the asynchronous path as there are less than $n-f$ replicas left on the synchronous path due to quorum intersection, so no progress can be made on the synchronous path (see algorithm~\ref{algo:fallbacksteadyalgorithm} line~\ref{S:line:9}) and all replicas will timeout (see algorithm~\ref{algo:fallbacksteadyalgorithm} line~\ref{S:line:27}). As a result, at least $n-f$ correct replicas will broadcast their <timeout> message and all replicas will enter the asynchronous path.
    
    Upon entering the asynchronous path, each replica creates a asynchronous block with height $1$ and broadcasts it (see algorithm~\ref{algo:fallbackalgorithm} line~\ref{A:line:6}-\ref{A:line:7}). Since we use perfect point-to-point links, eventually all the height $1$ blocks sent by the $n-f$ correct replicas will be received by each replica in the asynchronous path. At least $n-f$ correct replicas will send them <vote-async> messages if the rank of the height $1$ block is greater than the rank of the replica (see algorithm~\ref{algo:fallbackalgorithm} line~\ref{A:line:9}-\ref{A:line:10}). To ensure liveness for the replicas that have a lower rank, the algorithm allows catching up, so that nodes will adopt whichever height $1$ block which received $n-f$ <vote-async> arrives first. Upon receiving the first height $1$ block with $n-f$ <vote-async> messages, each replica will send a height $2$ asynchronous block (see algorithm~\ref{algo:fallbackalgorithm} line~\ref{A:line:15}-\ref{A:line:19}), which will be eventually received by all the replicas in the asynchronous path. Since the height $2$ block proposed by any block passes the rank test for receiving a <vote-async>, eventually at least $n-f$ height $2$ blocks get $n-f$ <vote-async> (see algorithm~\ref{algo:fallbackalgorithm} line~\ref{A:line:10}). Hence, eventually at least $n-f$ replicas send the <asynchronous-complete> message (see algorithm~\ref{algo:fallbackalgorithm} line~\ref{A:line:22}), and exit the asynchronous path.
\end{proof}

\begin{theorem}\label{th:10}
    With probability $p>\frac{1}{2}$, at least one replica commits an elected-asynchronous block after exiting the asynchronous path.
\end{theorem}
\begin{proof}
    Let leader $L$ be the output of the common-coin-flip($v$) (see algorithm~\ref{algo:fallbackalgorithm} line~\ref{A:line:25}). A replica commits a block during the view change if the asynchronous-complete block from $L$ is among the first $n-f$ <asynchronous-complete> messages received during the asynchronous mode (see algorithm~\ref{algo:fallbackalgorithm} line~\ref{A:line:26}), which happens with probability at least greater than $\frac{1}{2}$. Hence with probability no less than $\frac{1}{2}$, each replica commits a chain in a given asynchronous phase.
\end{proof}

\begin{theorem}\label{th:11}
    A majority of replicas keep committing new blocks with high probability.
\end{theorem}
\begin{proof}
    We first prove this theorem for the basic case where all honest replicas start the protocol with $v = 0$. If at least $n-f$ replicas eventually enter the asynchronous path, by theorem~\ref{th:9}, they eventually all exit the asynchronous path, and a new block is committed by at least one replica with probability no less than $\frac{1}{2}$. According to the asynchronous-complete step (see algorithm~\ref{algo:fallbackalgorithm} line~\ref{A:line:33}), all nodes who enter the asynchronous path enter view $v = 1$ after exiting the asynchronous path. If at least $n-f$ replicas never set $isAsync$ to true, this implies that the sequence of blocks produced in view $1$ is infinite. By Theorem~\ref{th:2}, the blocks have consecutive round numbers, and thus a majority replicas keep committing new blocks.

Now assume the theorem~\ref{th:11} is true for view $v = 0,..., k-1$. Consider the case where at least $n-f$ replicas enter the view $v = k$. By the same argument for the $v = 0$ base case, $n-f$ replicas either all enter the asynchronous path commits a new block with $\frac{1}{2}$ probability, or keeps committing new blocks in view $k$. Therefore, by induction, a majority replicas keep committing new blocks with high probability.
\end{proof}

\begin{theorem}\label{th:12}
    Each client transaction is eventually committed.
\end{theorem}
\begin{proof}
    If each replica repeatedly keeps proposing the client transactions until they become committed, then eventually each client transaction gets committed according to theorem~\ref{th:11}.
\end{proof} 

\end{document}